\documentclass[journal]{IEEEtran}
\usepackage{graphicx}
\usepackage{amssymb,amsmath,amsthm}
\usepackage{mathtools}
\usepackage[ruled,vlined]{algorithm2e}
\usepackage{xcolor}
\usepackage{url}
\usepackage{enumitem}

\usepackage{pgfplots}
\pgfplotsset{compat=1.16}
\usepgfplotslibrary{statistics}

\usepackage{pgfplotstable}
\usepackage{booktabs}
\usepackage{array}
\usepackage{colortbl}
\usepackage{subcaption}
\usepackage{tcolorbox}

\usepackage[style=ieee]{biblatex}
\newcommand{\papertitle}{Coupled regularized sample covariance matrix
estimator for multiple classes}

\newcommand{\keywords}{Covariance matrix estimation, regularization,
shrinkage, elliptical distribution, regularized discriminant analysis.}

\scriptsize

\usepackage[pdftex,
            pdfauthor={Elias Raninen and Esa Ollila},
            pdftitle={\papertitle},
            pdfkeywords={\keywords},
            pdfcreator={pdflatex}]{hyperref}

\newcommand\numberthis{\addtocounter{equation}{1}\tag{\theequation}}

\theoremstyle{remark}
\newcounter{ctheorem}
\newtheorem{theorem}[ctheorem]{Theorem}

\newcounter{clemma}
\newtheorem{lemma}[clemma]{Lemma}

\newcounter{cproposition}
\newtheorem{proposition}[cproposition]{Proposition}

\newcommand{\mat}{\mathbf}

\renewcommand{\vec}{\mathbf}

\DeclareMathOperator{\E}{\mathbb{E}}

\DeclareMathOperator{\tr}{tr}

\newcommand{\A}{\mat{A}}
\newcommand{\B}{\mat{B}}

\newcommand{\I}{\mat{I}}

\renewcommand{\S}{\mat{S}}
\newcommand{\T}{\mat{T}}

\newcommand{\x}{\vec{x}}

\renewcommand{\a}{\vec{a}}

\newcommand{\bmu}{\boldsymbol{\mu}}
\newcommand{\bo}{\boldsymbol{0}}

\newcommand{\covm}{\boldsymbol{\Sigma}}
\newcommand{\M}{\boldsymbol{\Sigma}}

\newcommand{\real}{\mathbb{R}}

\newcommand{\Fro}{\mathrm{F}}
\newcommand\norm[1]{\left\lVert{#1}\right\rVert}

\newcommand{\argmin}{\operatornamewithlimits{arg~min\ }}

\usepackage{mathtools}
\DeclarePairedDelimiterX{\ip}[2]{\langle}{\rangle_\Fro}{#1, #2}
\newcommand\fn[1]{\left\lVert{#1}\right\rVert_\Fro}

\newcommand{\MSE}{\mathrm{MSE}}
\newcommand{\NMSE}{\mathrm{NMSE}}

\begin{document}
\title{\papertitle}
\author{Elias~Raninen,~\IEEEmembership{Student Member,~IEEE,}
	Esa~Ollila,~\IEEEmembership{Senior Member,~IEEE}%
	\thanks{E. Raninen and E. Ollila are with the Department of Signal
	Processing and Acoustics, Aalto University, P.O. Box 15400, FI-00076
    Aalto, Finland. The work was supported in part by the Academy of
    Finland Grant 298118. 
    }
	}

\maketitle
\begin{abstract}
	The estimation of covariance matrices of multiple classes with limited
	training data is a difficult problem. The sample covariance matrix
	(SCM) is known to perform poorly when the number of variables is large
	compared to the available number of samples. In order to reduce the
	mean squared error (MSE) of the SCM, regularized (shrinkage) SCM
	estimators are often used. In this work, we consider regularized SCM
	(RSCM) estimators for multiclass problems that couple together two
	different target matrices for regularization: the pooled (average) SCM
	of the classes and the scaled identity matrix. Regularization toward
	the pooled SCM is beneficial when the population covariances are
	similar, whereas regularization toward the identity matrix guarantees
	that the estimators are positive definite. We derive the MSE optimal
	tuning parameters for the estimators as well as propose a method for
	their estimation under the assumption that the class populations
	follow (unspecified) elliptical distributions with finite fourth-order
	moments. The MSE performance of the proposed coupled RSCMs are
	evaluated with simulations and in a regularized discriminant analysis
	(RDA) classification set-up on real data. The results based on three
	different real data sets indicate comparable performance to
	cross-validation but with a significant speed-up in computation time. 
\end{abstract}


\IEEEpeerreviewmaketitle

\section{Introduction}\label{sec:introduction}
\IEEEPARstart{A}{n} increasingly common scenario in modern supervised learning
problems is that the dimension $p$ of the data is large compared to the number
of available training samples $n$ or exceed it multifold ($p \gg n$). Such
scenarios are commonly referred to as high-dimensional or insufficient sample
support problems. In this paper, we address the problem of high-dimensional
covariance matrix estimation in a multiclass setup, where there are $K$
different classes or populations, each comprising $n_k$, $k=1,\ldots,K$,
independent and identically distributed (i.i.d.) $p$-dimensional samples.
Estimates of the covariance matrices are needed in many multivariate analysis
problems, such as in principal component analysis and canonical correlation
analysis~\cite{bilodeau1999theory}, discriminant analysis~\cite{Friedman1989},
Gaussian mixture models (GMM)~\cite{Halbe2013}, as well as in many engineering
applications, e.g., in array signal processing~\cite{du_fully_2010},
genomics~\cite{schafer_shrinkage_2005}, portfolio optimization in
finance~\cite{ledoit_improved_2003}, and graphical models~\cite{zhang2013}. The
success of the analysis is often directly lined with the accuracy of the
estimated covariance matrix.

The covariance matrix of class $k \in \{1,\ldots,K\}$ is defined as
\begin{equation*}
	\M_k 
	= 
	\mathbb{E}[(\x_{ik}-\bmu_k)( \x_{ik}-\bmu_k)^\top], 
\end{equation*}
where $\x_{ik}$ denotes the $i$th sample from class $k$ and $\bmu_k =
\mathbb{E}[\x_{ik}]$ is the mean of class $k$. The conventional estimate for
the covariance matrix is the unbiased sample covariance matrix (SCM) defined
for class $k$ by
\begin{equation*}
	\S_k =
	\frac{1}{n_k-1}\sum_{i=1}^{n_k}{ (\x_{ik}-\overline\x_k)
	(\x_{ik}-\overline\x_k)^\top}, 
\end{equation*}
where $\overline \x_k = (1/n_k)\sum_i \x_{ik}$ is the sample mean of class
$k$.  In high-dimensional settings, the SCM is known to work poorly due to its
high variability. Furthermore, if $n_k < p$, then the SCM is singular, and
hence, its inverse cannot be computed. 

Better estimators can be developed by using regularization, where the key idea
is to shift or shrink the estimator toward a predetermined target or model.
This can significantly decrease the variance of the estimator and improve the
overall performance by reducing its mean squared error (MSE). This phenomenon
can be understood via the following well-known bias-variance decomposition of
the MSE. For an estimator $\hat \M_k$ of $\M_k$, the MSE can be written as
\begin{align*}
    \MSE(\hat\M_k)
    &= \E\left[\norm{\M_k - \hat \M_k}^2_\Fro\right]
    \label{eq:L}
    \\
    &=
    \E\left[\norm{\hat\M_k - \E[\hat\M_k]}_\Fro^2\right]
    + \norm{\M_k - \E[\hat \M_k]}^2_\Fro,
\end{align*}
where the first term on the right-hand side is the variance and the second
term is the squared bias of the estimator. Since the SCM is unbiased, its MSE
is equal to its variance. By using a regularized SCM (RSCM), however, it is
possible to reduce the MSE significantly at the cost of introducing some bias.

In general, regularization combines an unstructured estimate with a predefined
model or target. The target in regularization can be decided based on a prior
knowledge, assumptions, or on a property, which we want to enforce.
Regularization can be accomplished different ways. For instance, the SCM can
be combined linearly with a target matrix as, e.g.,
in~\cite{ledoit_improved_2003},~\cite{ledoit_well_conditioned_2004},~\cite{ledoit_honey},~\cite{schafer_shrinkage_2005},~\cite{chen_shrinkage_2010},~\cite{Li2017},~\cite{Maio2019},
and~\cite{Ollila2019}. The approaches for parameter tuning differs between the
methods. A popular approach taken in~\cite{ledoit_well_conditioned_2004} is
based on estimating the asymptotically optimal (in terms of minimizing the
MSE) tuning parameters. Under additional distributional assumptions, the
method of~\cite{ledoit_well_conditioned_2004} has later been improved for
Gaussian samples in~\cite{chen_shrinkage_2010} and elliptically distributed
samples in~\cite{Ollila2019}. Other approaches for tuning parameter selection
are for example the expected likelihood
approach~\cite{Besson2013,Abramovich2015}. Alternatively, the eigenvalues of
the SCM can be transformed non-linearly toward a specific structure as
in~\cite{ledoit_spectrum_2015} and~\cite{ledoit_analytical_2020}. Multiple
targets can also be used. For instance, in~\cite{Ikeda2016}, a double
shrinkage covariance matrix estimator was considered, which shrinks
simultaneously toward a spherical matrix and a diagonal matrix. A somewhat
related RSCM formulation had been proposed in~\cite{Halbe2013} in the setting
of Gaussian mixture models. In~\cite{Lancewicki2014}, the estimator
of~\cite{ledoit_well_conditioned_2004} was extended to multiple simultaneous
target matrices, which satisfy a certain target structure.
In~\cite{Tong2018}, a linear multi-target shrinkage covariance matrix
estimator was proposed, which optimizes the tuning parameters using
low-complexity leave-one-out cross-validation. There are also methods tailored
for multiclass problems. For instance,~\cite{Besson2020} considered covariance
matrix estimation from two possibly mismatched data sets using the maximum
likelihood principle. In our previous work~\cite{raninen2020linear}, linear
pooling of SCMs was considered and was applied to portfolio optimization. In
the context of discriminant analysis classification, a popular RSCM was
proposed in~\cite{Friedman1989}. Inspired by~\cite{Friedman1989}, this paper
focuses on two specific target matrices: the pooled SCM and the spherical
matrix.

The rest of the paper is organized as follows. In
Section~\ref{sec:motivation}, we give some background and motivation for the
proposed estimator. In Section~\ref{sec:derivationalphabeta}, we analyze some
properties of the proposed estimator as well as discuss the optimization of
the tuning parameters. In Section~\ref{sec:estimationofparameters}, we show
how the theoretical tuning parameters can be estimated in practice when the
unknown class populations follow unspecified elliptically symmetric
distributions. Section~\ref{sec:practicalimplementation} discusses some
practical considerations and how to use the method in choosing the tuning
parameters in RDA. In Section~\ref{sec:simulations}, synthetic simulation
studies are conducted in order to assess the MSE performance of the estimator.
The method is then compared to cross-validation in choosing the tuning
parameters for RDA classification using three different real data sets.
Lastly, Section~\ref{sec:conclusion} concludes.

\emph{Notation}:
Throughout the paper, unless stated otherwise, all norms are Frobenius norms
defined by $\|\mat A\|_\Fro^2 = \ip{\A}{\A} = \tr(\mat{A}^\top \mat{A})$,
where the inner product of two matrices (of appropriate dimensions) $\A$ and
$\B$ is defined by $\ip{\A}{\B} = \tr(\mat A^\top \mat B)$. For any square
matrix $\A$, we frequently use the notation $\I_{\A} = (1/p) \tr(\A) \I$ and
$\A^\I = \A - \I_{\A}$. For a vector $\a \in \real^p$, the Euclidean norm is
defined as $\|\a\| = \sqrt{\a^\top \a}$. We define $\real_{\geq 0} = \{a \in
\real : a \geq 0\}$. For a scalar variable $a$, we use the following shorthand
notation $\partial_a = \partial/\partial a$ for the partial derivative with
respect to $a$. For scalars $a < b$ and $c$, we define the \emph{clamp} or
\emph{clip} function $[c]_a^b = \max\{a,\min\{b,c\}\}$, which projects $c$
onto to the interval $[a,b]$. Lastly, $\text{Unif}\{a,b\}$ and
$\text{Unif}(a,b)$ denote the discrete and continuous univariate uniform
distributions on the set $\{a,a+1,\ldots,b-1,b\}$ and on the open interval
$(a,b)$, respectively.

\section{Background and motivation}\label{sec:motivation}
In multiclass problems, when the classes can be assumed to have a similar
covariance structure, it is beneficial to shrink the individual class
covariance matrix estimates toward the pooled (average) SCM of the classes,
\begin{equation}
    \S = \sum_{k=1}^K \pi_k \S_k,
    \label{eq:pooledS}
\end{equation}
where
\begin{equation*}
    \pi_k = \frac{n_k}{n_1+n_2+\cdots+n_K}.
\end{equation*}
For example, the methods proposed in~\cite{Friedman1989},~\cite{Greene1989},
and~\cite{Rayens1991} used the convex combination 
\begin{equation}
    \hat \M_k(\beta)
    = \beta \S_k + (1-\beta)\S,
    \label{eq:partiallypooledSCM}
\end{equation}
where $\beta \in [0,1]$, as an estimate for the class covariance matrix. 
The methods in~~\cite{Friedman1989}, \cite{Greene1989}, and~\cite{Rayens1991}
were designed for discriminant analysis (DA) classification in which the
problem is to classify a new sample $\x$ to one of the $K$ classes using the
\emph{discriminant rule}
\begin{equation}
    \hat k = \argmin_{k \in \{1,\ldots,K\}}
    {(\x-\hat\bmu_k)}^\top {\hat\M_k}^{-1} (\x-\hat\bmu_k)
    + \log|\hat\M_k|,
    \label{eq:DArule}
\end{equation}
where $\hat\bmu_k$ and $\hat\M_k$ denote estimates of the mean and the
covariance matrix of class $k$, respectively. Different DA methods differ in
the approach used to estimate the class means and class covariance matrices.
For example, \emph{quadratic discriminant analysis} (QDA) uses sample means
and SCMs in~\eqref{eq:DArule}. In \emph{linear discriminant analysis} (LDA)
the pooled SCM $\S$ in~\eqref{eq:pooledS} is used for all classes. If the true
covariance matrices are equal, $\M_1=\M_2=\ldots=\M_K \equiv \M$, then LDA is
well justified since the pooled SCM is an unbiased estimator of $\M$, i.e.,
$\E[\S] = \sum_{j} \pi_j \E[\S_j] = (\sum_j \pi_j)(\M) = \M$. However, even in
the case that the true population covariance matrices differ substantially,
LDA often outperforms QDA when the sample size is small compared to the
dimension~\cite{Greene1989},~\cite{Hastie2009}. This is due to the high
variance of the class SCMs compared to the pooled SCM. The partially pooled
estimator~\eqref{eq:partiallypooledSCM} includes QDA and LDA as special cases
when $\beta=1$ and $\beta=0$, respectively.

In a very high-dimensional case, when $p > \sum_j n_j$, the partially pooled
estimator~\eqref{eq:partiallypooledSCM} will no longer be positive definite.
In \emph{regularized discriminant analysis}~\cite{Friedman1989} (RDA), this
problem is solved by coupling the pooled estimator with a spherical target
matrix, that is, by regularization toward a scaled identity matrix via
\begin{align*}
    \hat\M_k(\alpha,\beta) = 
    \alpha \hat\M_k(\beta)
    +
    (1-\alpha)
    \I_{\hat \M_k(\beta)},
    \label{eq:RSCMRDA}
    \numberthis
\end{align*}
where $\hat \M_k(\beta)$ is given in~\eqref{eq:partiallypooledSCM} and
$\I_{\hat \M_k(\beta)} = (1/p)\tr(\hat \M_k(\beta))\I$.

The success of regularization depends on proper selection of the tuning
parameters. In classification problems, cross-validating the classification
error is the standard method for choosing the tuning parameters.
Cross-validation can, however, be computationally very costly for large data
sets. In certain high-dimensional binary classification settings, it is
possible to leverage on results from random matrix theory in order to estimate
the (asymptotically) optimal tuning parameters, which minimize the
classification error~\cite{yangRegularizedDiscriminantAnalysis2018,
elkhalilLargeDimensionalStudy2020}.

The main contribution of this work is to develop a computationally efficient
method for choosing the tuning parameters in a multiclass covariance matrix
estimation setting, where there can be more than two classes, and where the
RSCMs are defined in~\eqref{eq:RSCMRDA}. Specifically, we find estimates of the
tuning parameters that minimize the class-specific MSE:
\begin{equation}
    (\alpha_k^\star,\beta_k^\star)
    =
    \argmin_{\alpha,\beta \in [0,1]}
    \E\left[\norm{\hat \M_k (\alpha,\beta) - \M_k}^2_\Fro\right]
    \label{eq:alphabetaoptimal}
\end{equation} 
for each population $k = 1,\ldots,K$. The expressions for the optimal tuning
parameters will depend on unknown population parameters since the MSE
expression involves the unknown covariance matrix. However, we show that the
MSE can be estimated fairly easily by assuming that the class populations
follow unspecified elliptically symmetric distributions.

It is worth noting that, in RDA~\cite{Friedman1989} the tuning parameters are
common across the classes, whereas in this paper we use class-specific tuning
parameters. However, in cases when it is useful to use common tuning
parameters, they can easily be acquired by averaging as explained in
Subsection~\ref{sec:averaging}. Hence, our proposed method can be used to
obtain tuning parameters for the original RDA~\cite{Friedman1989} framework,
for which there already exist widely established toolboxes in programming
languages such as R (see
e.g.,~\cite{klaRpackage,roeverKlaRClassificationVisualization2020}). Lastly, an important
distinction to~\cite{Friedman1989} is that our method is not only targeted for
classification problems but is suitable for other applications as well. For
example,~\cite{Besson2020} considered the problem of estimating a covariance
matrix from two data sets, where the population covariance matrix of the first
data set is different but close to the population covariance matrix of the
second data set. This type of problems are encountered in radar processing as
well as in hyperspectral imaging applications, where the additional data sets
may have been acquired with slightly different measurement configurations, and
hence, have slightly different population parameters.

The current paper extends our earlier preliminary work in~\cite{Raninen2018},
which considered the estimation of the MSE optimal tuning parameters for the
partially pooled estimator in~\eqref{eq:partiallypooledSCM}. Here, we consider
the more general estimator in~\eqref{eq:RSCMRDA}, which includes additional
shrinkage toward the scaled identity matrix, and is therefore applicable also
in the cases when $p > \sum_j n_j$. 

\section{Estimator}\label{sec:derivationalphabeta}
Let us first consider four special cases of the estimator~\eqref{eq:RSCMRDA}:
\begin{enumerate} 
    \item[(C1)] \emph{The unpooled regularized SCM estimator} omits the
	pooled SCM and only shrinks toward the scaled identity matrix: 
        \[
	    \hat \M_k(\alpha_k,\beta_k=1) 
	    = \alpha_k \S_k + (1-\alpha_k) \I_{\S_k}.
        \]
	This type of shrinkage is typically considered in single class
	covariance matrix estimation (see
	e.g.,~\cite{ledoit_well_conditioned_2004} and~\cite{Ollila2019}).

    \item[(C2)] \emph{The partially pooled estimator} omits regularization
	toward the scaled identity and only shrinks toward the pooled SCM:
        \[
            \hat \M_k(\alpha_k=1,\beta_k) 
            = \hat \M_k(\beta_k) = \beta_k \S_k + (1-\beta_k) \S.
        \]

    \item[(C3)] \emph{The fully pooled estimator} uses the pooled SCM for
	every class $k$ and shrinks it toward the scaled identity matrix:
        \[
            \hat \M_k(\alpha_k,\beta_k=0) 
			= \alpha_k \S + (1-\alpha_k) \I_{\S}.
        \]
        Such shrinkage can be considered if all classes have an identical
        distribution.

    \item[(C4)] \emph{The scaled identity estimator} uses the partially pooled
	estimator to scale the identity matrix:
        \[
            \hat \M_k(\alpha_k=0,\beta_k) 
			= \I_{(\beta_k \S_k + (1-\beta_k) \S)}.
        \]
\end{enumerate} 

Since it is clear that the tuning parameters are class-specific, we drop the
subscripts from $\alpha_k$ and $\beta_k$ and denote them from now on simply by
$\alpha$ and $\beta$.

Figure~\ref{fig:abplane} depicts the theoretical normalized MSE (NMSE),
$\NMSE(\hat \M_k(\alpha,\beta)) = \MSE(\hat
\M_k(\alpha,\beta))/\|\M_k\|_\Fro^2$, of the estimator~\eqref{eq:RSCMRDA} as a
function of the tuning parameters $(\alpha,\beta)$. The figure corresponds to
setup A of the numerical study of Section~\ref{sec:simulations}. In the
figure, the small gray dots depict the estimated tuning parameters (showing
400 realizations of the 4000 Monte Carlo trials). The blue square
$({\color{blue} \blacksquare})$ denotes the mean $(m_\alpha, m_\beta)$ of the
estimated tuning parameters over the Monte Carlo runs. The optimal tuning
parameter pair $(\alpha^\star,\beta^\star)$ is denoted by the black triangle
($\blacktriangle$). The special cases (C1--C4) correspond to the edges of the
$(\alpha,\beta)$-plane.

\def\aloptfour{3.097890e-01}
\def\beoptfour{2.048585e-01}
\def\nmseofoptfour{2.994326e-01}
\def\almeanfour{3.343960e-01}
\def\bemeanfour{2.257217e-01}
\def\nmseofmeanfour{3.002223e-01}
\def\Cxxyyfour{1.126043e+00}
\def\Cxxyfour{1.363044e-01}
\def\Cxxfour{9.699572e-01}
\def\Cyyfour{1.887015e-03}
\def\Cxyfour{-1.876994e-01}
\def\Cxfour{-6.090920e-01}
\def\Cyfour{1.680789e-05}
\def\Cfour{3.996509e-01}

\pgfplotstableread{results/AR1-estimatedtuningparametersofclass4.dat}\estimatedtuningparameters
\begin{figure}[th]
    \centering
    \begin{tikzpicture}
	\begin{axis}[colormap/cool,
	    xlabel=$\alpha$,
	    ylabel=$\beta$,
	    zlabel={NMSE},
	    width=\linewidth,
	    view={20}{20}]

	    \addplot3[surf,
	    shader=flat,
	    samples=25,
	    domain=0:1,
	    y domain=0:1,
	    ] {\Cxxyyfour*x^2*y^2 + \Cxxyfour*x^2*y +
	    \Cxxfour*x^2 + \Cyyfour*y^2 + \Cxyfour*x*y +
	    \Cxfour*x + \Cyfour*y + \Cfour};

	    \addplot3[only marks, mark size=0.2, mark=*, gray=0.5] table
	    [x=alpha, y=beta, z=nmse] {\estimatedtuningparameters};

	    \addplot3[only marks,mark size=2.0,blue,mark=square*, fill] coordinates
	    {(\almeanfour,\bemeanfour,\nmseofmeanfour)} node[anchor = north
	    west] {\footnotesize$(m_\alpha,m_\beta)$};

	    \addplot3[only marks,mark size=2.0,black,mark=triangle*]
	    coordinates {(\aloptfour,\beoptfour,\nmseofoptfour)} node[anchor =
	    south east]
	    {\footnotesize$(\alpha^\star,\beta^\star)$};

	    \addplot[] coordinates {(0.9,0.9)} node[] {$\displaystyle \S_k$};
	    \addplot[] coordinates {(0.9,0.1)} node[] {$\displaystyle \S$};
	    \addplot[] coordinates {(0.1,0.1)} node[] {$\displaystyle \I_{\S}$}; 
	    \addplot[] coordinates {(0.1,0.9)} node[] {$\displaystyle
	    \I_{\S_k}$};
	\end{axis}
    \end{tikzpicture}
    \caption{\small The theoretical NMSE as a function of the tuning
    parameters $(\alpha, \beta)$ for class 4 of the setup A in
    Section~\ref{sec:simulations}.}
    \label{fig:abplane}
\end{figure}
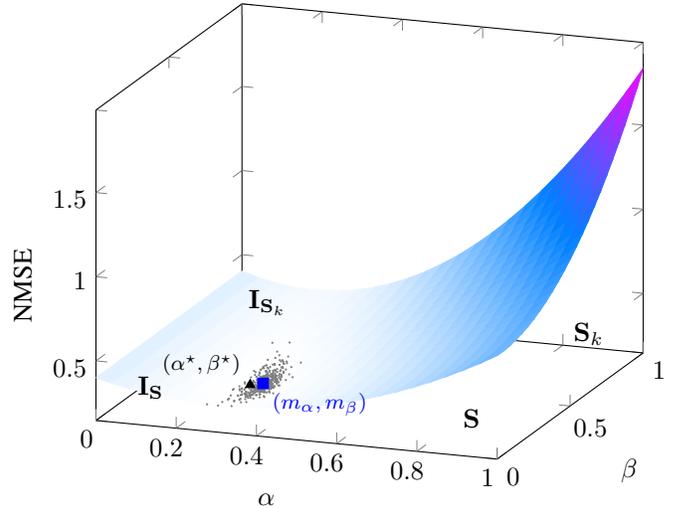
As can be observed from Figure~\ref{fig:abplane}, the optimal tuning parameter
pair is in this case closest to $\I_\S$. Hence, in this particular case it
would be a good strategy to use the pooled SCM and regularize it toward a
scaled identity matrix, i.e., to use the fully pooled estimator (C3). As can be
noted, the proposed method is able to automatically choose the tuning parameters
in a near optimal way.

Let us now give an expression for the MSE of the estimator.

\begin{theorem}\label{thm:poly}
    The MSE of the estimator~\eqref{eq:RSCMRDA} is a bivariate polynomial of
    the form 
    \begin{align*}
	\MSE(\hat \M_k(\alpha,\beta))
	&= 
	\alpha^2 \beta^2 C_{22}
	+ \alpha^2 \beta C_{21}
	+ \alpha^2 C_{20}
	+ \beta^2 C_{02}
	\\&\qquad
	+ \alpha \beta C_{11}
	+ \alpha C_{10}
	+ \beta C_{01}
	+ C_{00},
	\numberthis
	\label{eq:MSEalphabetapolynomial}
    \end{align*}
    where the coefficients, $C_{ij}$, depend on the scalars $\tr(\M_j)$,
    $\E[\|\S_j\|_\Fro^2]$, $\E[\|\I_{\S_j}\|_\Fro^2]$, and $\ip{\M_i}{\M_j}$.
\end{theorem}
\begin{proof} 
    See Appendix~\ref{app:MSE} for the proof and the expressions for the
    coefficients $C_{ij}$.
\end{proof}
The estimation of the coefficients $C_{ij}$ is deferred to
Section~\ref{sec:estimationofparameters}. Given estimates of the coefficients,
the critical points of~\eqref{eq:MSEalphabetapolynomial} can be solved
numerically (see Appendix~\ref{app:MSE}).

A useful property of the MSE polynomial in~\eqref{eq:MSEalphabetapolynomial}
is that given a fixed value of either $\alpha$ or $\beta$, the optimization of
the remaining tuning parameter is a convex problem. This is known as
biconvexity.
\begin{theorem}\label{thm:biconvex}
    The MSE~\eqref{eq:MSEalphabetapolynomial} is a biconvex function. The
    optimal tuning parameter $\alpha^\star$ given a fixed value of $\beta \in
    [0,1]$ is
    \begin{align*}
	\alpha^\star
	&=
	\left[
	-\frac{1}{2}
	\frac{\beta C_{11} + C_{10}}{\beta^2 C_{22} + \beta C_{21} + C_{20}}
	\right]_0^1.
	\label{eq:alphagivenbeta}
	\numberthis
    \end{align*}
    Likewise, the optimal tuning parameter $\beta^\star$ given a fixed value
    of $\alpha \in [0,1]$ is
    \begin{align*}
	\beta^\star 
	&=
	\left[
	-\frac{1}{2}\frac{\alpha^2 C_{21} + \alpha C_{11}
	+ C_{01}}{\alpha^2 C_{22} + C_{02}}
	\right]_0^1.
	\label{eq:betagivenalpha}
	\numberthis
    \end{align*}
\end{theorem}
\begin{proof} 
    See Appendix~\ref{app:biconvex}.
\end{proof}
The optimal tuning parameter for the special cases (C1--C4) can be solved
using~\eqref{eq:alphagivenbeta} and~\eqref{eq:betagivenalpha} of
Theorem~\ref{thm:biconvex}.

\subsection{Partially pooled estimator without identity shrinkage}

Let us next highlight certain properties of the partially pooled
estimator~\eqref{eq:partiallypooledSCM} of (C2) corresponding to the case
$\hat\M_k(\alpha=1,\beta)$. The propositions below provide additional insights
on the optimal parameter $\beta^\star$ in this case. 
\begin{proposition}\label{proposition:betalessthan1}
    The optimal tuning parameter $\beta^\star$ given $\alpha = 1$
    satisfies $\beta^\star < 1$.
\end{proposition}
\begin{proof}
    The proof follows from showing that the numerator
    of~\eqref{eq:betagivenalpha} is always less than the denominator, which is
    always positive. By setting $\alpha=1$, after some algebra, we have that
    $-1/2$ times the numerator, i.e., $(-1/2)(C_{21}+C_{11}+C_{01})$, is equal to
    $\E[\ip{\S-\S_k}{\S - \M_k}]$ and the denominator ($C_{22} + C_{02}$) is
    equal to $\E[\fn{\S - \S_k}^2]$ (See Appendix~\ref{app:MSE}). Subtracting
    the numerator from the denominator, we get
    \begin{align*}
        &\E[\fn{\S-\S_k}^2] - \E[\ip{\S - \S_k}{\S - \M_k}]
        \\
        &=\E[\ip{\S - \S_k}{\M_k - \S_k}]
        \\
        &= (1-\pi_k)(\E[\fn{\S_k}^2] - \fn{\M_k}^2)
        \\
        &= (1-\pi_k)\E[\fn{\S_k - \M_k}^2] > 0,
    \end{align*}
    which holds almost surely (a.s.) for any continuous distribution. 
\end{proof}

This has the important implication that in a multiclass problem it is always
possible to reduce the MSE of the SCM by using regularization toward the pooled
SCM.

In the special case that all of the population covariance matrices are equal,
we would like to only use the pooled SCM. This is exactly what happens as is
shown in the next proposition.
\begin{proposition}\label{proposition:equalcovariances}
    Consider multiple populations generated by the same distribution such that
    they have equal covariance matrices $\M_1 = \ldots =\M_K$ and equal sample
    sizes $\pi_1=\ldots=\pi_K$. Given that $\alpha=1$ then $\beta^\star =
    0$.
\end{proposition}
\begin{proof}
    The proof follows from showing that if the true population covariance
    matrices and the sample sizes are equal, then the numerator
    of~\eqref{eq:betagivenalpha} is zero and the denominator is non-zero.
    Observe that if $\M_k = \M_j$ and $\pi_k = \pi_j$, $\forall j$, then we have
    $\E[\fn{\S_k}^2] = \E[\fn{\S_j}^2]$ and $\ip{\M_j}{\M_k} = \fn{\M_k}^2$,
    $\forall j$. Hence, opening up the numerator $\E[\ip{\S - \S_k}{\S - \M_k}]$
    yields
    \begin{align*}
        &\E[\fn{\S}^2] - \E[\ip{\S}{\M_k}] - \E[\ip{\S_k}{\S}] +
        \E[\ip{\S_k}{\M_k}]
        \\
        &= \E[\fn{\S}^2] - \E[\ip{\S_k}{\S}]
        \\
        &= (\sum_j \pi_j^2 - \pi_k) \E[\fn{\S_k}^2] 
        + ( \sum_{i\neq j} \pi_i \pi_j - \sum_{j,j\neq k}\pi_j )\fn{\M_k}^2
        \\
        &= 0,
    \end{align*}
    which follows from $\sum_j \pi_j^2 = K\pi_k^2 = (K \pi_k) \pi_k = \pi_k$,
    $\sum_{i \neq j} \pi_i \pi_j = K(K-1) \pi_k^2 = (K\pi_k)(K-1)\pi_k =
    1-\pi_k$, and $\sum_{j,j \neq k} \pi_j = (K-1) \pi_k = 1-\pi_k$.
    Furthermore, using this result we see that the denominator, $\E[\fn{\S -
    \S_k}^2]$, is a.s. positive. Opening up the denominator, we have
    \begin{align*}
        & \E[\fn{\S}^2] + \E[\fn{\S_k}^2] - 2\E[\ip{\S}{\S_k}]
        \\
        &= 
        \E[\fn{\S_k}^2] - \E[\ip{\S}{\S_k}]
        \\
        &= (1-\pi_k) \E[\fn{\S_k}^2] - \sum_{j,j\neq k} \pi_j \ip{\M_k}{\M_j}
        \\
        &= (1-\pi_k) \left(\E[\fn{\S_k}^2] - \fn{\M_k}^2\right) > 0 ~\text{a.s.}
    \end{align*}
\end{proof}

\subsection{Streamlined analytical estimator}\label{sec:streamlined}
An alternative estimator to~\eqref{eq:RSCMRDA} can be obtained if we change
the $\alpha$-regularization target so that the estimator becomes
\begin{equation}
    \tilde \M_k(\alpha,\beta) = \alpha \hat \M_k(\beta) + (1-\alpha) \I_{\T},
    \label{eq:RSCMstreamlined}
\end{equation}
where $\T \in \{\S_k, \S\}$ and $\hat \M_k(\beta)$ is defined
in~\eqref{eq:partiallypooledSCM}. This simplifies the expression for the MSE
and allows for an analytical solution for the tuning parameters.
In~\cite{Halbe2013} and~\cite{Ikeda2016}, somewhat similar formulations were
used with different targets in a GMM and a single class covariance matrix
estimation setting, respectively. Note that the difference
between~\eqref{eq:RSCMstreamlined} and~\eqref{eq:RSCMRDA} is in the scale of the
identity target. The trace of~\eqref{eq:RSCMstreamlined} (sum of the
eigenvalues) is dependent on $\alpha$, whereas in~\eqref{eq:RSCMRDA} this is not
the case. However, when $\tr(\I_\T) \approx \tr(\I_{\hat \M_k(\beta)})$, the
performance of the two estimators is expected to be similar. The MSE
of~\eqref{eq:RSCMstreamlined} is given in the next theorem.

\begin{theorem}\label{thm:polysimple}
    The theoretical MSE of the (streamlined analytical)
    estimator~\eqref{eq:RSCMstreamlined} is a bivariate polynomial of the form
    \begin{align*}
	\MSE(\tilde \M_k(\alpha,\beta))
	&= 
	\alpha^2 \beta^2 B_{22}
	+ \alpha^2 \beta B_{21}
	+ \alpha^2 B_{20}
	\\&\qquad
	+ \alpha \beta B_{11}
	+ \alpha B_{10}
	+ B_{00}.
    \end{align*}
    The coefficients $B_{ij}$ depend on the scalars $\tr(\M_j)$,
    $\E[\|\S_j\|_\Fro^2]$, $\E[\|\I_{\S_j}\|_\Fro^2]$, and $\ip{\M_i}{\M_j}$. If
    $(\alpha^\star,\beta^\star) \in (0,1) \times (0,1)$, the optimal tuning
    parameters $(\alpha^\star,\beta^\star)$ minimizing the MSE are
    \begin{align*}
	\alpha^\star
	= \frac{2 B_{10} B_{22} - B_{11}B_{21}}
	{B_{21}^2 - 4B_{20} B_{22}}
	~\text{and}~
	\beta^\star
	= \frac{2 B_{11} B_{20} - B_{10} B_{21}}
	{2B_{10} B_{22} - B_{11} B_{21}}. 
    \end{align*}
    Otherwise, the optimal parameters are on the boundary of the feasible set
    $[0,1] \times [0,1]$, and are given by one of the following options
    \begin{enumerate}
	\item[i)]
	    $\alpha^\star = \bigg[-\dfrac{1}{2}\dfrac{B_{10}}{B_{20}}
	    \bigg]_0^1$ and $\beta^\star = 0$,
	\item[ii)]
	    $\alpha^\star = \bigg[-\dfrac{1}{2}\dfrac{B_{10} + B_{11}}{B_{22}
	    + B_{21} + B_{20}} \bigg]_0^1$ and $\beta^\star = 1$,
	\item[iii)]
	    $\alpha^\star = 1$ and $\beta^\star =
	    \left[-\dfrac{1}{2}\dfrac{B_{21} + B_{11}}{B_{22}}\right]_0^1$,
	\item[iv)] 
	    $\alpha^\star = 0$, which implies $\tilde \M = \I_\T$ and that the
	    MSE does not depend on $\beta$.
    \end{enumerate} 
\end{theorem}
\begin{proof}
    See Appendix~\ref{app:streamlined} for the proof and the expressions for the
    coefficients $B_{ij}$. 
\end{proof}

\section{Estimating the tuning parameters}\label{sec:estimationofparameters} 

The MSE and the optimal tuning parameters in~Theorem~\ref{thm:poly}
and~Theorem~\ref{thm:polysimple} depend on the coefficients $C_{ij}$ and
$B_{ij}$, which in turn are functions of the unknown scalars 
\begin{equation} \label{eq:unknown_params} 
    \tr(\M_j),~\E[\|\S_j\|_\Fro^2],~\E[\|\I_{\S_j}\|_\Fro^2]
    ~\text{and}~\ip{\M_i}{\M_j}. 
\end{equation} 
Hence, estimates of the optimal tuning parameters can be formed by a plug-in
method, i.e., replacing the above unknown parameters by their estimates. Such
estimates are constructed in this section under the very general assumption
that the samples are generated from unspecified elliptical distributions with
finite fourth-order moments.

Elliptical distributions are a location-scale family of spherical
distributions generalizing many common distributions, such as the multivariate
normal distribution (MVN) and Student's $t$-distribution to name a few. The
probability density function (p.d.f.) of an elliptically distributed sample
$\x$ from class $k$ is up to a normalizing constant of the form
\begin{align*}
    |\M_k|^{-1/2} g_k \left( (\x - \bmu_k)^\top \M_k^{-1} (\x - \bmu_k)
    \right),
\end{align*}
where $\bmu_k$ and $\M_k$ are the mean and the positive definite covariance
matrix of the distribution, respectively. The function $g_k : \real_{\geq 0}
\rightarrow \real_{\geq 0}$ is called the \emph{density generator}, which
determines the distribution. For example, $g_k(t) = \exp(-t/2)$ defines the
MVN distribution. For references about elliptical distributions see,
e.g.,~\cite{fang2018symmetric} and~\cite{frahm2004generalized}.

Estimates of the unknown scalars, $\tr(\M_k)$ and $\fn{\M_k}^2$, can be
obtained by finding estimates of the \emph{scale} parameter $\eta_k$ and the
\emph{sphericity} parameter $\gamma_k$ defined as
\begin{equation*}
    \eta_k = \frac{\tr(\M_k)}{p}
    ~\text{and}~
    \gamma_k = \frac{\fn{\M_k}^2}{\fn{\I_{\M_k}}^2} \in [1,p].
\end{equation*}
Observe that $\fn{\I_{\M_k}}^2 = p\eta_k^{2}$ and $\fn{\M_k}^2 = p \gamma_k
\eta_k^2$. Under the ellipticity assumption, the unknowns
$\E[\|\S_k\|_\Fro^2]$ and $\E[\|\I_{\S_k}\|_\Fro^2]$ are functions of the
scale, sphericity, and the elliptical kurtosis as shown in the next lemma. The
elliptical kurtosis is defined as $\kappa_k = (1/3)\text{kurt}(x_i)$, where
$\text{kurt}(x_i)$ is the excess kurtosis of any marginal variable $x_i$ of a
random vector from the $k$th class.
\begin{lemma}\label{thm:EtrSk2}
    \cite[Lemma 2.]{Ollila2019}
    Consider an elliptical distribution with covariance matrix $\M_k$ and
    finite fourth-order moments. Then, 
    \begin{align*}
	\E[\|\S_k\|_\Fro^2] 
	&= p\eta_k^2\left(\tau_{1k} p+(1+\tau_{1k}+\tau_{2k}) \gamma_k \right)
	~\text{and}
	\\
	\E[\|\I_{\S_k}\|_\Fro^2]
	&= \eta_k^2 \left( (1+\tau_{2k})p + 2\tau_{1k} \gamma_k\right) 
	,
    \end{align*}
    where $\tau_{1k} = 1/(n_k-1) + \kappa_k/n_k$, and $\tau_{2k} =
    \kappa_k/n_k$.
\end{lemma}

In the following subsections, we show how to estimate the unknown parameters,
$\eta_k$, $\gamma_k$, $\kappa_k$, and $\ip{\M_i}{\M_j}$. We use the same
approach as in~\cite{raninen2020linear}.

\subsection{Estimation of the elliptical kurtosis}
The elliptical kurtosis of class $k$ is estimated by the sample average
\begin{equation*}
    \hat \kappa_k 
    = 
    \max \Bigg\{ 
    \frac{1}{3 p} \sum_{j=1}^p \hat k_j 
    ~,~
    - \frac{2}{p+2}
    \Bigg\},
\end{equation*}
where $\hat k_j = m^{(4)}_j/ \big(m_j^{(2)} \big)^2 - 3$ is the sample
estimate of the kurtosis of the $j$th variable (of class $k$), and $m^{(q)}_j
= \frac{1}{n_k} \sum_{i=1}^{n_k}((\x_{ik})_j - (\bar\x_k)_j)^q$ is the $q$th
order sample moment. Here, the notation $(\cdot)_j$ picks the $j$th variable
from the vector. The maximum constraint ensures that $\hat \kappa_k$ respects
the theoretical lower bound of $-2/(p+2)$~\cite{bentler1986greatest}. Note
that, the finite fourth-order moments are required in order to have a finite
elliptical kurtosis.

\subsection{Estimation of the sphericity}
It would be natural to use the SCM in order to develop an estimator for
$\gamma_k$. However,
following~\cite{zhang2016automatic},~\cite{Ollila2019},~\cite{Raninen2018},
and~\cite{raninen2020linear}, we use the following simple and well performing
estimator
\begin{equation*}
    \hat \gamma_k = \left[\frac{p n_k}{n_k-1}\Big(\|\tilde\S_k\|_\Fro^2
    - \frac{1}{n_k}\Big)\right]_1^p,
\end{equation*}
which uses the \emph{sample spatial sign covariance matrix} (SSCM) defined as
\begin{equation*}
    \tilde\S_k 
    = \frac{1}{n_k} \sum_{i=1}^{n_k}
    \frac{(\x_{ik} - \bmu_k){(\x_{ik}- \bmu_k)}^\top}
    {\norm{\x_{ik}- \bmu_k}^2},
\end{equation*}
where the mean $\bmu_k$ is estimated by the spatial median~\cite{Brown1983}
$\tilde \bmu_k = \arg \min\limits_{\bmu} \sum_{i=1}^{n_k} \norm{\x_{ik} -
\bmu}$. Particularly, in~\cite{Ollila2019}, the SSCM based sphericity estimator
was compared to a SCM based sphericity estimator, when sampling from elliptical
distributions. The SSCM based estimator performed better in all cases except
when sampling from a MVN distribution. Furthermore, when $\bmu$ is known,
in~\cite{raninen2020linear} it was shown that the expectation of the estimator
is asymptotically unbiased as the dimension grows if $\gamma_k/p \to 0$ as $p
\to \infty$. This assumption was shown to hold, for example, for the first order
autoregressive (AR(1)) covariance matrix but not for the compound symmetry (CS)
covariance matrix (see Section~\ref{sec:simulations} for the definitions of
these covariance matrix structures)~\cite{raninen2020linear}.

\subsection{Estimation of the scale and the inner products}
The scale $\eta_k$ as well as the inner products $\ip{\M_i}{\M_j}$, for $i
\neq j$, can be estimated by using the SCMs, i.e., by $\hat \eta_k =
\tr(\S_k)/p$ and $\ip{\S_i}{\S_j}$, respectively. However, regarding the inner
products $\ip{\M_i}{\M_j}$, we use the same SSCM based estimator as
in~\cite{raninen2020linear}, which is $\hat\eta_i \hat\eta_j p^2 \ip{\tilde
\S_i}{\tilde \S_j}$. When $i=j$, we estimate
$\ip{\M_k}{\M_k}=\|\M_k\|_\Fro^2$, by $p\hat \gamma_k \hat \eta_k^2$.

\section{Practical considerations}\label{sec:practicalimplementation}
In this section, we first discuss how to compute the estimates of the optimal
tuning parameters in practise. Then, we discuss regularizing the tuning
parameters themselves, which will in effect make the method compatible with
RDA~\cite{Friedman1989}.

\subsection{Computation of the optimal tuning
parameters}\label{sec:computationoftuningparameters}

A straightforward way of solving for the optimal tuning
parameters~\eqref{eq:alphabetaoptimal} for the proposed
estimator~\eqref{eq:RSCMRDA} is to form a two-dimensional grid of
$(\alpha,\beta) \in [0,1] \times [0,1]$ and choose the point, which yields the
minimum (estimated) MSE in~\eqref{eq:MSEalphabetapolynomial}. The solution can
further be fine-tuned via an alternating convex minimization by iterating
between~\eqref{eq:alphagivenbeta} and~\eqref{eq:betagivenalpha}. Convergence of
the iterations is addressed in Appendix~\ref{app:biconvex}. This method of
finding the optimal tuning parameters was used in the simulations of
Section~\ref{sec:simulations}, where the method is denoted by \textbf{POLY} (due
to the polynomial structure of the MSE in Theorem~\ref{thm:poly}). In the
simulations, the particular $(\alpha,\beta)$-grid was constructed from all pairs
in the set $\{0,0.05,\ldots,1\} \ni (\alpha,\beta)$. 

An alternative way of solving for the optimal tuning parameters is to find the
critical points of the MSE polynomial~\eqref{eq:MSEalphabetapolynomial}, which
is addressed in Appendix~\ref{app:MSE}.

Given estimates for the coefficients $C_{ij}$, evaluating the
equations~\eqref{eq:MSEalphabetapolynomial},~\eqref{eq:alphagivenbeta},
and~\eqref{eq:betagivenalpha} is computationally very light. Most of the
computation is due to estimating the coefficients and not in optimizing the
tuning parameters.

The streamlined analytical estimator~\eqref{eq:RSCMstreamlined} with $\T =
\S$, which is discussed in Section~\ref{sec:streamlined}, is denoted by
\textbf{POLYs}. In this case, the tuning parameters are found using the
plug-in estimates of the optimal values stated in
Theorem~\ref{thm:polysimple}. 

\subsection{Averaging of the tuning parameters}\label{sec:averaging}
The estimation accuracy of the optimal tuning parameters $(\alpha^\star_k,
\beta^\star_k)_{k=1}^K$ depend on the estimates of the statistical population
parameters discussed in Section~\ref{sec:estimationofparameters}. If the
estimation of these parameters is difficult for some data sets, which can
happen if the elliptical distribution assumption does not fit the data well,
regularization of the tuning parameters themselves can be useful. A possible
way to accomplish this is to average the tuning parameters of the classes by
using $\bar{\hat \alpha} = \frac{1}{K}\sum \hat \alpha_k$ and $\bar{\hat
\beta} =\frac{1}{K} \sum \hat \beta_k$ in place of the class-specific tuning
parameters. Averaging the estimated class tuning parameters is reasonable as
long as the optimal tuning parameters are not too different from each other,
which is the case when the class covariance matrices are similar. If the
population covariance matrices are very different from each other, averaging
the estimated tuning parameters might degrade the performance. However, by
using $(\bar{\hat \alpha}, \bar{\hat \beta})$ as common tuning parameters for
all classes in~\eqref{eq:RSCMRDA}, the resulting RSCM estimator will be
compatible with RDA proposed in~\cite{Friedman1989}. Thus, this method can be
used instead of cross-validation as an alternative way for choosing the tuning
parameters for RDA. In Section~\ref{sec:simulations}, we illustrate that this
approach enables a significant speed up in the computation of RDA with no
effective loss in performance.

The averaged version of the estimator~\eqref{eq:RSCMRDA}, which uses the mean
of the estimated tuning parameters for all classes is denoted by
\textbf{POLY-Ave}. Regarding the streamlined analytical estimator
of~\eqref{eq:RSCMstreamlined}, the averaged version is denoted by
\textbf{POLYs-Ave}. 

\section{Numerical studies}\label{sec:simulations}

\pgfplotstableread{results/MIXED-estimatedtuningparametersofclass1.dat}\classone
\pgfplotstableread{results/MIXED-estimatedtuningparametersofclass2.dat}\classtwo
\pgfplotstableread{results/MIXED-estimatedtuningparametersofclass3.dat}\classthree
\pgfplotstableread{results/MIXED-estimatedtuningparametersofclass4.dat}\classfour
\def\aloptone{4.706538e-01}
\def\beoptone{4.419414e-01}
\def\nmseofoptone{3.306596e-01}
\def\almeanone{4.930519e-01}
\def\bemeanone{4.224059e-01}
\def\nmseofmeanone{3.310449e-01}
\def\Cxxyyone{1.231391e+00}
\def\Cxxyone{-2.682701e-01}
\def\Cxxone{7.690333e-01}
\def\Cyyone{1.138469e-03}
\def\Cxyone{-3.874665e-01}
\def\Cxone{-6.674481e-01}
\def\Cyone{-3.158870e-04}
\def\Cone{5.279421e-01}

\def\alopttwo{4.508684e-01}
\def\beopttwo{3.699300e-01}
\def\nmseofopttwo{3.455084e-01}
\def\almeantwo{4.824149e-01}
\def\bemeantwo{3.770785e-01}
\def\nmseofmeantwo{3.465112e-01}
\def\Cxxyytwo{1.351423e+00}
\def\Cxxytwo{-1.482376e-01}
\def\Cxxtwo{7.690333e-01}
\def\Cyytwo{1.741647e-03}
\def\Cxytwo{-3.874665e-01}
\def\Cxtwo{-6.674481e-01}
\def\Cytwo{2.872914e-04}
\def\Ctwo{5.279421e-01}

\def\aloptthree{6.754718e-01}
\def\beoptthree{4.194770e-01}
\def\nmseofoptthree{3.274568e-01}
\def\almeanthree{5.887563e-01}
\def\bemeanthree{4.044308e-01}
\def\nmseofmeanthree{3.337467e-01}
\def\Cxxyythree{8.737208e-01}
\def\Cxxythree{1.045536e-01}
\def\Cxxthree{5.442932e-01}
\def\Cyythree{8.241782e-04}
\def\Cxythree{-5.664716e-01}
\def\Cxthree{-7.646321e-01}
\def\Cythree{-2.051607e-04}
\def\Cthree{6.658950e-01}

\def\aloptfour{6.629747e-01}
\def\beoptfour{3.448618e-01}
\def\nmseofoptfour{3.478985e-01}
\def\almeanfour{5.808437e-01}
\def\bemeanfour{3.608357e-01}
\def\nmseofmeanfour{3.522181e-01}
\def\Cxxyyfour{9.594023e-01}
\def\Cxxyfour{1.902351e-01}
\def\Cxxfour{5.442932e-01}
\def\Cyyfour{1.254738e-03}
\def\Cxyfour{-5.664716e-01}
\def\Cxfour{-7.646321e-01}
\def\Cyfour{2.253995e-04}
\def\Cfour{6.658950e-01}

\usepgfplotslibrary{groupplots}
\begin{figure*}[th]
    \centering
    \begin{tikzpicture}
	\begin{groupplot}[group style={group size=4 by 1},
	    tick label style={font=\footnotesize},
	    height=5cm,
	    width=5cm,
	    colormap/cool,
	    xlabel=$\alpha$,
	    ylabel=$\beta$,
	    view={20}{20},
	    ]
	    \nextgroupplot[title={Class 1},
	    zlabel={NMSE},
	    ]
	    \addplot3[surf,
	    shader=flat,
	    samples=25,
	    domain=0:1,
	    y domain=0:1,
	    ] 
	    {\Cxxyyone*x^2*y^2 + \Cxxyone*x^2*y + \Cxxone*x^2 +
	    \Cyyone*y^2 + \Cxyone*x*y + \Cxone*x + \Cyone*y + \Cone};

	    \addplot3[only marks, mark size=0.2, mark=*, gray=0.5] table
	    [x=alpha, y=beta, z=nmse] {\classone};

	    \addplot3[only marks,mark size=2.0,blue,mark=square*, fill] coordinates
	    {(\almeanone,\bemeanone,\nmseofmeanone)};

	    \addplot3[only marks,mark size=2.0,black,mark=triangle*] coordinates
	    {(\aloptone,\beoptone,\nmseofoptone)};

	    \nextgroupplot[title={Class 2},]
	    \addplot3[surf,
	    shader=flat,
	    samples=25,
	    domain=0:1,
	    y domain=0:1,
	    ] 
	    {\Cxxyytwo*x^2*y^2 + \Cxxytwo*x^2*y + \Cxxtwo*x^2 +
	    \Cyytwo*y^2 + \Cxytwo*x*y + \Cxtwo*x + \Cytwo*y +
	    \Ctwo};

	    \addplot3[only marks, mark size=0.2, mark=*, gray=0.5] table
	    [x=alpha, y=beta, z=nmse] {\classtwo};

	    \addplot3[only marks,mark size=2.0,blue,mark=square*, fill] coordinates
	    {(\almeantwo,\bemeantwo,\nmseofmeantwo)};

	    \addplot3[only marks,mark size=2.0,black,mark=triangle*] coordinates
	    {(\alopttwo,\beopttwo,\nmseofopttwo)};

	    \nextgroupplot[title={Class 3},]
	    \addplot3[surf,
	    shader=flat,
	    samples=25,
	    domain=0:1,
	    y domain=0:1,
	    ] 
	    {\Cxxyythree*x^2*y^2 + \Cxxythree*x^2*y +
	    \Cxxthree*x^2 + \Cyythree*y^2 + \Cxythree*x*y +
	    \Cxthree*x + \Cythree*y + \Cthree};

	    \addplot3[only marks, mark size=0.2, mark=*, gray=0.5]
	    table [x=alpha, y=beta, z=nmse] {\classthree};

	    \addplot3[only marks,mark size=2.0,blue,mark=square*, fill] coordinates
	    {(\almeanthree,\bemeanthree,\nmseofmeanthree)};

	    \addplot3[only marks,mark size=2.0,black,mark=triangle*] coordinates
	    {(\aloptthree,\beoptthree,\nmseofoptthree)};

	    \nextgroupplot[title={Class 4},]
	    \addplot3[surf,
	    shader=flat,
	    samples=25,
	    domain=0:1,
	    y domain=0:1,
	    ] 
	    {\Cxxyyfour*x^2*y^2 + \Cxxyfour*x^2*y + \Cxxfour*x^2 +
	    \Cyyfour*y^2 + \Cxyfour*x*y + \Cxfour*x + \Cyfour*y + \Cfour};

	    \addplot3[only marks, mark size=0.2, mark=*, gray=0.5] table
	    [x=alpha, y=beta, z=nmse] {\classfour};

	    \addplot3[only marks,mark size=2.0,blue,mark=square*, fill] coordinates
	    {(\almeanfour,\bemeanfour,\nmseofmeanfour)};

	    \addplot3[only marks,mark size=2.0,black,mark=triangle*] coordinates
	    {(\aloptfour,\beoptfour,\nmseofoptfour)};

	\end{groupplot}
    \end{tikzpicture}
    \caption{\small The theoretical NMSE as a function of the tuning
    parameters $(\alpha, \beta)$ for the setup C (mixed setup). The black
    triangle ($\blacktriangle$) denotes the optimal tuning parameter pair. The
    first 400 realizations of the estimated tuning parameters over the 4000
    Monte Carlo trials are shown as the gray dots. The blue square
    $({\color{blue} \blacksquare})$ denotes the mean of the estimated tuning
    parameters.}\label{fig:mixednmse}
\end{figure*}
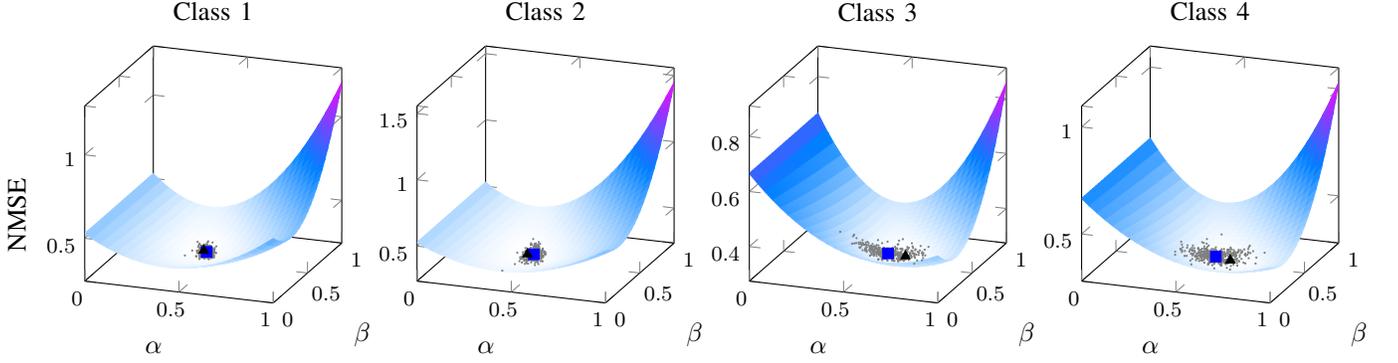

In this section, the performance of the proposed estimators is reported and
compared against other competing methods. First, in
Subsection~\ref{sec:syntheticsimulations}, the MSE performance of the method
is examined with synthetic simulations. Then, in
Subsection~\ref{sec:RDAsimulations} the method is applied to RDA
classification and compared to cross-validation in choosing the tuning
parameters. 

\subsection{Synthetic simulations}\label{sec:syntheticsimulations}

\begin{table}
    \centering
    \caption{\small Empirical NMSE ($\times 10$) of setups A, B, C, and D.
    The sample standard deviation ($\times 10$) is given in the
    parenthesis.}\label{table:NMSE}
    \begin{scriptsize}
	\tabcolsep=0.12cm
	\begin{tabular}{lcccc|c}
	    \toprule
	    & Class 1 & Class 2 & Class 3 & Class 4 & Sum\\
	    \midrule
	    \multicolumn{6}{l}{\underline{\emph{Setup A. AR(1)}}} \\
	SCM & 115.4 \scriptsize{(73.8)} & 51.5 \scriptsize{(32.8)} & 29.7 \scriptsize{(12.8)} & 18.4 \scriptsize{(8.2)} & 214.9 \scriptsize{(81.5)} \\
	POOL & 12.5 \scriptsize{(3.1)} & 10.6 \scriptsize{(2.8)} & 8.9 \scriptsize{(2.3)} & 7.6 \scriptsize{(1.9)} & 39.6 \scriptsize{(10.1)} \\
	PPOOL & 12.4 \scriptsize{(2.6)} & 10.4 \scriptsize{(2.3)} & 8.8 \scriptsize{(1.8)} & 7.5 \scriptsize{(1.6)} & 39.1 \scriptsize{(8.2)} \\
	Ell-RSCM & 1.1 \scriptsize{(0.4)} & 1.7 \scriptsize{(0.2)} & 2.6 \scriptsize{(0.1)} & 3.3 \scriptsize{(0.1)} & 8.6 \scriptsize{(0.5)} \\
	LOOCV1 & 10.6 \scriptsize{(20.9)} & 9.6 \scriptsize{(19.5)} & 8.8 \scriptsize{(6.9)} & 8.3 \scriptsize{(5.6)} & 37.3 \scriptsize{(28.6)} \\
	LOOCV2 & 1.1 \scriptsize{(0.6)} & 1.4 \scriptsize{(0.2)} & 2.1 \scriptsize{(0.1)} & 3.0 \scriptsize{(0.1)} & 7.6 \scriptsize{(0.7)} \\
	LIN1 & 5.3 \scriptsize{(0.5)} & 4.8 \scriptsize{(0.3)} & 4.6 \scriptsize{(0.3)} & 4.7 \scriptsize{(0.3)} & 19.4 \scriptsize{(1.1)} \\
	LIN2 & 0.9 \scriptsize{(0.5)} & 1.3 \scriptsize{(0.2)} & 2.1 \scriptsize{(0.1)} & 3.0 \scriptsize{(0.1)} & 7.3 \scriptsize{(0.6)} \\
	POLY & 0.9 \scriptsize{(0.3)} & 1.3 \scriptsize{(0.1)} & 2.1 \scriptsize{(0.1)} & 3.0 \scriptsize{(0.1)} & 7.2 \scriptsize{(0.5)} \\
	POLYs & 0.8 \scriptsize{(0.1)} & 1.3 \scriptsize{(0.1)} & 2.1 \scriptsize{(0.1)} & 3.0 \scriptsize{(0.1)} & 7.1 \scriptsize{(0.3)} \\
	POLY-Ave & 1.0 \scriptsize{(0.5)} & 1.4 \scriptsize{(0.3)} & 2.1 \scriptsize{(0.1)} & 3.1 \scriptsize{(0.1)} & 7.7 \scriptsize{(0.7)} \\
	POLYs-Ave & 1.0 \scriptsize{(0.3)} & 1.4 \scriptsize{(0.2)} & 2.1 \scriptsize{(0.1)} & 3.1 \scriptsize{(0.1)} & 7.6 \scriptsize{(0.5)} \\
	    \midrule
	    \multicolumn{6}{l}{\underline{\emph{Setup B. Compound symmetry}}} \\
	SCM & 14.7 \scriptsize{(11.4)} & 3.6 \scriptsize{(3.3)} & 1.5 \scriptsize{(1.0)} & 0.8 \scriptsize{(0.8)} & 20.6 \scriptsize{(11.9)} \\
	POOL & 10.5 \scriptsize{(4.9)} & 1.8 \scriptsize{(1.3)} & 0.5 \scriptsize{(0.3)} & 0.7 \scriptsize{(0.4)} & 13.5 \scriptsize{(6.1)} \\
	PPOOL & 6.6 \scriptsize{(3.4)} & 1.4 \scriptsize{(1.0)} & 0.5 \scriptsize{(0.3)} & 0.6 \scriptsize{(0.4)} & 9.0 \scriptsize{(4.0)} \\
	Ell-RSCM & 5.7 \scriptsize{(1.3)} & 2.9 \scriptsize{(1.1)} & 1.4 \scriptsize{(0.7)} & 0.7 \scriptsize{(0.5)} & 10.7 \scriptsize{(1.9)} \\
	LOOCV1 & 1.7 \scriptsize{(4.0)} & 1.0 \scriptsize{(1.8)} & 0.7 \scriptsize{(0.8)} & 0.7 \scriptsize{(0.7)} & 4.1 \scriptsize{(4.5)} \\
	LOOCV2 & 1.5 \scriptsize{(3.0)} & 0.9 \scriptsize{(1.2)} & 0.8 \scriptsize{(0.7)} & 0.7 \scriptsize{(0.6)} & 3.8 \scriptsize{(3.3)} \\
	LIN1 & 1.6 \scriptsize{(2.2)} & 0.8 \scriptsize{(0.9)} & 0.6 \scriptsize{(0.5)} & 0.6 \scriptsize{(0.4)} & 3.6 \scriptsize{(2.4)} \\
	LIN2 & 1.3 \scriptsize{(2.1)} & 0.7 \scriptsize{(0.9)} & 0.6 \scriptsize{(0.5)} & 0.6 \scriptsize{(0.4)} & 3.2 \scriptsize{(2.4)} \\
	POLY & 1.3 \scriptsize{(1.7)} & 0.7 \scriptsize{(0.7)} & 0.6 \scriptsize{(0.4)} & 0.6 \scriptsize{(0.4)} & 3.2 \scriptsize{(2.0)} \\
	POLYs & 1.3 \scriptsize{(1.7)} & 0.7 \scriptsize{(0.7)} & 0.6 \scriptsize{(0.4)} & 0.6 \scriptsize{(0.4)} & 3.1 \scriptsize{(2.0)} \\
	POLY-Ave & 3.3 \scriptsize{(2.2)} & 0.5 \scriptsize{(0.4)} & 0.8 \scriptsize{(0.4)} & 1.4 \scriptsize{(0.5)} & 6.0 \scriptsize{(1.9)} \\
	POLYs-Ave & 3.3 \scriptsize{(2.2)} & 0.5 \scriptsize{(0.4)} & 0.8 \scriptsize{(0.4)} & 1.4 \scriptsize{(0.5)} & 6.0 \scriptsize{(1.9)} \\
	    \midrule
	    \multicolumn{6}{l}{\underline{\emph{Setup C. Mixed}}} \\
	SCM & 12.1 \scriptsize{(1.9)} & 14.7 \scriptsize{(6.1)} & 8.6 \scriptsize{(1.5)} & 10.3 \scriptsize{(4.1)} & 45.6 \scriptsize{(7.7)} \\
	POOL & 6.3 \scriptsize{(0.9)} & 6.3 \scriptsize{(0.9)} & 4.5 \scriptsize{(0.5)} & 4.5 \scriptsize{(0.5)} & 21.5 \scriptsize{(2.0)} \\
	PPOOL & 5.4 \scriptsize{(0.5)} & 5.7 \scriptsize{(0.7)} & 3.8 \scriptsize{(0.4)} & 4.0 \scriptsize{(0.5)} & 19.0 \scriptsize{(1.5)} \\
	Ell-RSCM & 3.7 \scriptsize{(0.1)} & 3.9 \scriptsize{(0.2)} & 3.9 \scriptsize{(0.5)} & 4.1 \scriptsize{(0.6)} & 15.5 \scriptsize{(0.8)} \\
	LOOCV1 & 4.9 \scriptsize{(0.6)} & 5.6 \scriptsize{(2.4)} & 4.5 \scriptsize{(0.5)} & 5.0 \scriptsize{(1.8)} & 20.0 \scriptsize{(2.9)} \\
	LOOCV2 & 2.9 \scriptsize{(0.1)} & 2.9 \scriptsize{(0.1)} & 2.8 \scriptsize{(0.4)} & 2.8 \scriptsize{(0.4)} & 11.4 \scriptsize{(0.8)} \\
	LIN1 & 3.6 \scriptsize{(0.2)} & 3.6 \scriptsize{(0.2)} & 3.0 \scriptsize{(0.4)} & 3.1 \scriptsize{(0.4)} & 13.4 \scriptsize{(0.8)} \\
	LIN2 & 2.9 \scriptsize{(0.1)} & 2.9 \scriptsize{(0.2)} & 2.8 \scriptsize{(0.4)} & 2.8 \scriptsize{(0.4)} & 11.4 \scriptsize{(0.8)} \\
	POLY & 3.3 \scriptsize{(0.1)} & 3.4 \scriptsize{(0.2)} & 3.4 \scriptsize{(0.5)} & 3.5 \scriptsize{(0.5)} & 13.7 \scriptsize{(0.8)} \\
	POLYs & 3.3 \scriptsize{(0.1)} & 3.4 \scriptsize{(0.2)} & 3.4 \scriptsize{(0.5)} & 3.5 \scriptsize{(0.5)} & 13.7 \scriptsize{(0.8)} \\
	POLY-Ave & 3.3 \scriptsize{(0.1)} & 3.5 \scriptsize{(0.4)} & 3.4 \scriptsize{(0.4)} & 3.6 \scriptsize{(0.5)} & 13.9 \scriptsize{(0.8)} \\
	POLYs-Ave & 3.3 \scriptsize{(0.1)} & 3.5 \scriptsize{(0.4)} & 3.4 \scriptsize{(0.4)} & 3.6 \scriptsize{(0.5)} & 13.9 \scriptsize{(0.8)} \\
	    \midrule
	    \multicolumn{6}{l}{\underline{\emph{Setup D. Randomized}}} \\
	SCM & 20.6 \scriptsize{(72.9)} & 19.2 \scriptsize{(39.7)} & 19.6 \scriptsize{(48.7)} & 23.0 \scriptsize{(115.1)} & 82.4 \scriptsize{(149.3)} \\
	POOL & 41.7 \scriptsize{(110.9)} & 37.9 \scriptsize{(102.3)} & 38.0 \scriptsize{(86.3)} & 41.5 \scriptsize{(88.0)} & 159.1 \scriptsize{(254.9)} \\
	PPOOL & 10.1 \scriptsize{(16.0)} & 9.9 \scriptsize{(17.5)} & 9.9 \scriptsize{(17.6)} & 10.7 \scriptsize{(23.6)} & 40.6 \scriptsize{(37.1)} \\
	Ell-RSCM & 2.1 \scriptsize{(1.8)} & 2.1 \scriptsize{(1.7)} & 2.1 \scriptsize{(1.8)} & 2.1 \scriptsize{(3.6)} & 8.4 \scriptsize{(4.8)} \\
	LOOCV1 & 3.6 \scriptsize{(39.5)} & 2.7 \scriptsize{(4.7)} & 3.0 \scriptsize{(13.2)} & 3.2 \scriptsize{(16.2)} & 12.5 \scriptsize{(45.2)} \\
	LOOCV2 & 1.7 \scriptsize{(9.7)} & 1.5 \scriptsize{(1.7)} & 1.5 \scriptsize{(1.7)} & 1.6 \scriptsize{(2.6)} & 6.3 \scriptsize{(10.4)} \\
	LIN1 & 3.0 \scriptsize{(3.0)} & 2.9 \scriptsize{(2.3)} & 3.0 \scriptsize{(2.5)} & 3.1 \scriptsize{(4.1)} & 12.0 \scriptsize{(6.0)} \\
	LIN2 & 1.5 \scriptsize{(2.3)} & 1.4 \scriptsize{(1.4)} & 1.5 \scriptsize{(1.7)} & 1.5 \scriptsize{(3.6)} & 6.0 \scriptsize{(4.9)} \\
	POLY & 1.7 \scriptsize{(1.6)} & 1.6 \scriptsize{(1.5)} & 1.7 \scriptsize{(1.6)} & 1.7 \scriptsize{(3.5)} & 6.6 \scriptsize{(4.7)} \\
	POLYs & 1.7 \scriptsize{(1.6)} & 1.6 \scriptsize{(1.5)} & 1.7 \scriptsize{(1.6)} & 1.7 \scriptsize{(3.6)} & 6.6 \scriptsize{(4.9)} \\
	POLY-Ave & 6.2 \scriptsize{(16.3)} & 5.4 \scriptsize{(8.7)} & 5.8 \scriptsize{(10.8)} & 6.2 \scriptsize{(12.1)} & 23.5 \scriptsize{(22.6)} \\
	POLYs-Ave & 6.2 \scriptsize{(15.4)} & 5.4 \scriptsize{(8.7)} & 5.7 \scriptsize{(11.1)} & 6.3 \scriptsize{(24.4)} & 23.6 \scriptsize{(33.6)} \\
	    \bottomrule
	\end{tabular}
    \end{scriptsize}
\end{table}

\pgfplotsset{KappaEtaGammaStyle/.style={boxplot/draw direction=x,
ytick={1,2,3,4},
yticklabels={\makebox[2em]{$1$},\makebox[2em]{$2$},\makebox[2em]{$3$},\makebox[2em]{$4$}},
x tick label style={font=\small},
y tick label style={font=\small},
width=9cm,
height=3.0cm,
mark size=1,
cycle list={{blue,mark=+},
{blue,mark=o},
{red,mark=triangle},
{red,mark=square},
},
}}
\pgfplotsset{trCkCjStyle/.style={boxplot/draw direction=x,
ytick={1,2,3,4,5,6,7,8,9,10},
yticklabels={
    \makebox[2em]{$1 \& 1$}, 
    \makebox[2em]{$1 \& 2$}, 
    \makebox[2em]{$1 \& 3$},
    \makebox[2em]{$1 \& 4$},
    \makebox[2em]{$2 \& 2$},
    \makebox[2em]{$2 \& 3$},
    \makebox[2em]{$2 \& 4$},
    \makebox[2em]{$3 \& 3$},
    \makebox[2em]{$3 \& 4$},
    \makebox[2em]{$4 \& 4$}},
x tick label style={font=\small},
y tick label style={font=\small},
width=9cm,
height=5.0cm,
mark size=1,
cycle list={{blue,mark=+},
{gray,mark=o},
{gray,mark=o},
{gray,mark=o},
{blue,mark=o},
{gray,mark=o},
{gray,mark=o},
{red,mark=triangle},
{gray,mark=o},
{red,mark=square},
}}}
\usetikzlibrary{calc}
\begin{figure}[!ht]
    \centering
    \begin{tikzpicture}
	\begin{axis}[name=ploteta,KappaEtaGammaStyle,title={Estimates of $\eta_k$}]
	    \input{results/MIXED-ETA-addplot.tex}%
	\end{axis}
	\begin{axis}[name=plotkappa,at={($(ploteta.south)-(0,1.5cm)$)},anchor=north,KappaEtaGammaStyle,
	    title={Estimates of $\kappa_k$}]
	    \input{results/MIXED-KAPPA-addplot.tex}%
	\end{axis}
	\begin{axis}[name=plotgamma12,at={($(plotkappa.south)+(-2.0cm,-1.5cm)$)},anchor=north,
	    KappaEtaGammaStyle, title={Estimates of $\gamma_k$, $k=1,2$},
	    yticklabels={$1$,$2$},width=5cm]
	    \input{results/MIXED-GAMMA12-addplot.tex}%
	\end{axis}
	\begin{axis}[name=plotgamma34,at={($(plotkappa.south)+(2.0cm,-1.5cm)$)},anchor=north,
	    KappaEtaGammaStyle, title={Estimates of $\gamma_k$, $k=3,4$},
	    yticklabels={$3$,$4$},width=5cm,
	    cycle list={ {red,mark=triangle}, {red,mark=square},
	    },
	    ]
	    \input{results/MIXED-GAMMA34-addplot.tex}%
	\end{axis}
	\begin{axis}[name=plotcrossterms,at={($(plotkappa.south)-(0,4.5cm)$)},anchor=north,trCkCjStyle,
	    title={Estimates of $\ip{\M_i}{\M_j}$}]
	    \input{results/MIXED-trCkCj-addplot.tex}%
	\end{axis}
    \end{tikzpicture}
    \caption{\small Boxplots of the estimates of the parameters $\eta_k$,
    $\kappa_k$, $\gamma_k$, and the inner products $\ip{\M_i}{\M_j}$ for setup
    C. The black triangles ($\blacktriangle$) denote the true values.}
    \label{fig:etakappagamma}
\end{figure}

We evaluated the empirical NMSE performance of the proposed methods as well as
the accuracy of the plug-in estimates $\hat \eta_k$, $\hat \kappa_k$, $\hat
\gamma_k$, and the estimates of the inner products $\ip{\M_i}{\M_j}$. The
results were averaged over 4000 Monte Carlo trials. We simulated four
different setups: A, B, C, and D. In each setup, we generated $K=4$ classes,
each of dimension $p=200$. The data was generated from a multivariate
Student's $t_\nu$-distribution with various degrees of freedom $\nu$ and the
means of the classes were generated from the standard normal distribution
$\mathcal N_p(\bo,\I)$ and held fixed over the Monte Carlo trials for setups
A, B, and C. Regarding setup D, the mean as well as the other parameters were
randomly generated again for each Monte Carlo trial. The four different setups
were as follows.

\renewcommand{\labelenumi}{\emph{\Alph{enumi}}.}
\begin{enumerate}
    \item \emph{AR(1) process}
	\begin{itemize}
	    \item $n_1=25$, $n_2=50$, $n_3=75$, and $n_4 = 100$.
	    \item $\nu=8$ for all classes.
	    \item $(\M_k)_{ij} = \varrho_k^{|i-j|}$, where $\varrho_1 = 0.2$,
		$\varrho_2 = 0.3$, $\varrho_3 = 0.4$, and $\varrho_4 = 0.5$.
	\end{itemize}
    \item \emph{Compound symmetry}
	\begin{itemize}
	    \item $n_1=25$, $n_2=50$, $n_3=75$, and $n_4 = 100$.
	    \item $\nu=8$ for all classes.
	    \item $(\M_k)_{ii} = 1$ and $(\M_k)_{ij} = \varrho_k$ for $i\neq
		j$, where $\varrho_1=0.2$, $\varrho_2=0.3$, $\varrho_3=0.4$,
		and $\varrho_4=0.5$. 
	\end{itemize}
    \item \emph{Mixed structures}
	\begin{itemize}
	    \item $n_k = 100$, for all $k=1,2,3,4$.
	    \item $\nu_1=12$, $\nu_2 = 8$, $\nu_3 = 12$, and $\nu_4 = 8$.
	    \item $\M_1$ and $\M_2$ are as in setup A with
		$\varrho_1=\varrho_2=0.6$ and $\M_3$ and $\M_4$ are as in
		setup B with $\varrho_3=\varrho_4=0.1$.
	\end{itemize}
    \item \emph{Randomized}
	\\
	For each $k=1,2,3,4$ and each Monte Carlo trial:
	\begin{itemize}
	    \item $n_k = \text{Unif}\{10,200\}$.
	    \item $\nu_k = \text{Unif}\{5,12\}$.
	    \item $\bmu_k = \mathcal N_p(\bo,\I)$.
	    \item $\M_k$ either AR(1) or CS with equal probability and
		$\varrho_k = \text{Unif}(0,0.9)$.
	\end{itemize}
\end{enumerate}

In addition to the proposed estimators, in the simulations we include the
sample covariance matrix (SCM), the pooled SCM~\eqref{eq:pooledS} (POOL), the
partially pooled estimator~\eqref{eq:partiallypooledSCM} using the method
from~\cite{Raninen2018} (PPOOL) for choosing $\beta_k$, the method
from~\cite{Ollila2019} (ELL1), which is designed for single class covariance
matrix estimation and uses a convex combination $\beta_k \S_k +
(1-\beta_k)\hat\eta_k \I$. We also include the estimators proposed
in~\cite{raninen2020linear} (LIN1 and LIN2), which use a nonnegative linear
combination of the class SCMs. LIN2 differs from LIN1 in that it incorporates
additional shrinkage toward the identity matrix. In addition, we include the
shrinkage covariance matrix estimator proposed in~\cite{Tong2018} (LOOCV1 and
LOOCV2). The method is based on estimating the covariance matrix via a
nonnegative linear combination of the SCM and target matrices using
(low-complexity) leave-one-out cross-validation. In the simulations, for each
class $k$, LOOCV1 uses the pooled SCM and the identity matrix as target
matrices for regularizing the SCM of class $k$. Respectively, for each class
$k$, LOOCV2 uses the individual SCMs of the other classes as well as the
identity matrix as target matrices for regularizing the SCM of class $k$.

The empirical NMSE, $\textrm{Ave} \|\hat\covm_k -
\covm_k\|_{\Fro}^2/\norm{\covm_k}^2_{\Fro}$, is reported for each class in
Table~\ref{table:NMSE}. As can be noted, in the setups A, B, and C, the
proposed methods: POLY, POLYs, POLY-Ave, and POLYs-Ave yielded significantly
lower NMSE than the methods: SCM, POOL, PPOOL, ELL1, LIN1, and LOOCV1. The
only methods, which performed comparably well to POLY and POLYs were LIN2 and
LOOCV2. In setup A (AR(1) case), the best performing methods were POLY, POLYs,
and LIN2, which all had a similar NMSE. In setup B (CS case), the methods
POLY, POLYs performed best (LIN2 and LOOCV2 had a slightly higher NMSE or
higher standard deviation). In the setup C (mixed case), LIN2 and LOOCV2 had
similar performance and a slightly lower NMSE than the proposed methods. In
setup D (randomized case), LIN2 and LOOCV2 had slightly lower NMSE than POLY
and POLYs. However, LOOCV2 had over two times the standard deviation compared
to LIN2, POLY, and POLYs. In summary, when the covariance matrices had a
similar structure (setup A and setup B), the proposed POLY and POLYs methods
performed best, whereas when the covariance matrices had a different structure
LOOCV2 and LIN2 had a slight edge. The averaged versions of the proposed
methods, POLY-Ave and POLYs-Ave, performed well, when the covariance matrices
had similar structure (setup A and setup B), whereas in setup C and setup D,
their performance degraded compared to the methods using class-specific tuning
parameters. Lastly, the performances of POLY and POLYs were almost identical as
well as the performances of POLY-Ave and POLYs-Ave.

Figure~\ref{fig:etakappagamma} shows the boxplots of the estimated statistical
parameters $\eta_k$, $\kappa_k$, $\gamma_k$, and the inner products
$\ip{\M_i}{\M_j}$ of setup C, where classes 1 and 2 have an AR(1) covariance
structure and classes 3 and 4 have a CS structure. The median of the
estimated scales $\eta_k$ coincide with the true values (shown as black
triangles ($\blacktriangle$) in the plot). The elliptical kurtosis $\kappa_k$
was more difficult to estimate for heavier tailed distributions (classes 2 and
4 with $\nu=8$ degrees of freedom) than lighter tailed distributions (classes
1 and 3 with $\nu=12$ degrees of freedom). The sphericity $\gamma_k$ was well
estimated for classes 1 and 2, which both had the AR(1) covariance structure.
Regarding classes 3 and 4, which had the CS covariance structure, the
sphericity estimates were noticeably biased. This was most likely due to the
sphericity estimate not being asymptotically unbiased for the covariance matrix
with CS structure~\cite{raninen2020linear}. The inner products were well
estimated except for those pairs for which both covariance matrices had a CS
structure, i.e., $\ip{\M_i}{\M_j}$, for $(i,j) = \{(3,3),(3,4),(4,4)\}$. 

The estimated tuning parameters are shown in Figure~\ref{fig:mixednmse} for the
setup C. The figure depicts the theoretical NMSE of the estimator as a function
of the tuning parameters. The optimal tuning parameter pair
$(\alpha^\star,\beta^\star)$ is denoted by the black triangle
($\blacktriangle$). The small gray dots correspond to the estimated tuning
parameters. The first 400 estimates from the 4000 Monte Carlo trials are shown.
The mean of the estimated tuning parameters, denoted by the blue square
$({\color{blue} \blacksquare})$, was very close to the optimal value, although a
slight bias could be observed. The estimated tuning parameters were clustered
tightly around the optimal point, especially for classes 1 and 2 whose
covariance matrices had an AR(1) structure. Regarding classes 3 and 4, whose
covariance matrices had the CS structure, there was slightly more spread in the
estimates.

\subsection{Discriminant analysis}\label{sec:RDAsimulations}

We evaluated the performance of the proposed method \textbf{POLY-Ave} in
discriminant analysis classification. The implementation was written in
R~(programming language) using three different real data
sets: Sonar, Vowel, and Ionosphere, which were obtained
from~\cite{UCIMLrepository} via the R package
\emph{mlbench}~\cite{mlbenchpackage}. The features in the data sets, which
were constant for some of the classes were removed. Specifically, from the
Vowel data set, we removed the first feature, and from the Ionosphere data
set, we removed the first two features. The final specifications of the data
sets are given in Table~\ref{tab:datasets}. In our implementation, we used the
package \emph{SpatialNP}~\cite{SpatialNPpackage} for computing the SSCM more
efficiently, and the package \emph{tictoc}~\cite{tictocpackage} for recording
the computation times of the methods. 

In the simulations, we randomly selected a proportion of the samples as training
data for estimating the sample means and RSCMs. The remaining data was used as
test data to estimate the classification accuracy using the classification rule
given in~\eqref{eq:DArule}. The training set size is given as the number of
training samples divided by the total number of samples in the data set. So, the
training set size is given as a number between zero and one, where, e.g., $0.3$
corresponds to the case where $30 \%$ of the data is used as training data and
$70 \%$ of the data is used as test data to estimate the classification
accuracy. The results were averaged over 10 independent Monte Carlo repetitions
for each training set size.

For comparison to our proposed method, we included two different methods for
implementing cross-validation. The first method chooses the tuning parameters
from a two-dimensional grid of candidate parameter values as the minimizer of
the 5-fold or 10-fold cross-validated misclassification rate. In 10-fold
cross-validation the two-dimensional grid of tuning parameters scans all the
pairs in the set $\{0,0.125,0.25,0.375,0.5,0.625,0.75,0.875,1\} \ni
\alpha,\beta$. Correspondingly, in 5-fold cross-validation the set is
$\{0,0.25,0.5,0.75,1\} \ni \alpha,\beta$. This method is denoted by
\textbf{5-CV} and \textbf{10-CV} and was implemented using the packages:
\emph{caret}~\cite{kuhnCaretClassificationRegression2020} and
\emph{KlaR}~\cite{klaRpackage,roeverKlaRClassificationVisualization2020}. The
second cross-validation method is also based on selecting the tuning parameters
that minimize the cross-validated misclassification rate. However, instead of
using a predefined two-dimensional grid, the second method uses a
Nelder-Mead-algorithm included in the R package
\emph{KlaR}~\cite{klaRpackage,roeverKlaRClassificationVisualization2020}. This
method was implemented using both 5-fold and 10-fold cross-validation and is
denoted by \textbf{5-CV-NM} and \textbf{10-CV-NM}, respectively.

Figure~\ref{fig:classification} depicts the mean classification accuracy on
the test set as a function of the training set size as well as the median
computation time. It can be seen that all of the tested methods gave comparable
classification performance. Depending on the data set and training set size,
there were only minor differences. The computation time was, however,
substantially longer for the cross-validation based methods and it tended to
increase along with the training data size faster than for the proposed method.
The chosen tuning parameter values are shown in
Figure~\ref{fig:tuningparameters}. It can be seen that for the cross-validation
based methods the tuning parameter values varied a lot with different training
set sizes. The tuning parameters remained much more stable for our proposed
method.

\begin{table}
    \centering
    \caption{\small Specifications of the data sets.}\label{tab:datasets}
    \begin{tabular}{l | l l l}
	\toprule
	Data set & dimension $p$ & classes $K$ & samples $n_k$ \\
	\midrule
	Sonar & 59 & 2 & $n_1=111$, $n_2 = 97$ \\
	\midrule
	Vowel & 9 & 11 & $n_k = 90$ for all $k$ \\
	\midrule
	Ionosphere & 32 & 2 & $n_1=126$, $n_2 = 225$ \\
	\bottomrule
    \end{tabular}
\end{table}

\pgfplotstableread{results/classification_results_ionosphere_no_preprocessing.dat}\ionosphereresults
\pgfplotstableread{results/classification_results_sonar_no_preprocessing.dat}\sonarresults
\pgfplotstableread{results/classification_results_vowel_no_preprocessing.dat}\vowelresults
\begin{figure}[!ht]
    \begin{tikzpicture}
	\begin{axis}[name = rdaplot,
	    width = \linewidth,
	    title = {Classification accuracy},
	    xlabel = {training set size},
	    ylabel = {mean classification accuracy},
	    legend style={font=\tiny, draw=none},
	    legend pos = south west,
	    legend columns = 3,
	    ymin = 0.2,
	    ]
	    \addplot[red,mark=o,dashed] table[x=training_ratio,
	    y=rda.caret.cv5.acc.mean] {\sonarresults};
	    \addplot[blue,mark=o,dashed] table[x=training_ratio,
	    y=rda.caret.cv5.acc.mean] {\vowelresults};
	    \addplot[black,mark=o,dashed] table[x=training_ratio,
	    y=rda.caret.cv5.acc.mean] {\ionosphereresults};

	    \addplot[red,mark=square,dashed] table[x=training_ratio,
	    y=rda.caret.cv10.acc.mean] {\sonarresults};
	    \addplot[blue,mark=square,dashed] table[x=training_ratio,
	    y=rda.caret.cv10.acc.mean] {\vowelresults};
	    \addplot[black,mark=square,dashed] table[x=training_ratio,
	    y=rda.caret.cv10.acc.mean] {\ionosphereresults};

	    \addplot[red,mark=star] table[x=training_ratio,
	    y=rda.klaR.cv5.acc.mean] {\sonarresults};
	    \addplot[blue,mark=star] table[x=training_ratio,
	    y=rda.klaR.cv5.acc.mean] {\vowelresults};
	    \addplot[black,mark=star] table[x=training_ratio,
	    y=rda.klaR.cv5.acc.mean] {\ionosphereresults};

	    \addplot[red,mark=triangle] table[x=training_ratio,
	    y=rda.klaR.cv10.acc.mean] {\sonarresults};
	    \addplot[blue,mark=triangle] table[x=training_ratio,
	    y=rda.klaR.cv10.acc.mean] {\vowelresults};
	    \addplot[black,mark=triangle] table[x=training_ratio,
	    y=rda.klaR.cv10.acc.mean] {\ionosphereresults};

	    \addplot[red,mark=triangle*] table[x=training_ratio,
	    y=rda.pool.acc.mean] {\sonarresults};
	    \addplot[blue,mark=triangle*] table[x=training_ratio,
	    y=rda.pool.acc.mean] {\vowelresults};
	    \addplot[black,mark=triangle*] table[x=training_ratio,
	    y=rda.pool.acc.mean] {\ionosphereresults};
	    \legend{
		Sonar: 5-CV,
		Vowel: 5-CV,
		Ionosphere: 5-CV,
		Sonar: 10-CV,
		Vowel: 10-CV,
		Ionosphere: 10-CV,
		Sonar: 5-CV-NM,
		Vowel: 5-CV-NM,
		Ionosphere: 5-CV-NM,
		Sonar: 10-CV-NM,
		Vowel: 10-CV-NM,
		Ionosphere: 10-CV-NM,
		Sonar: POLY-Ave,
		Vowel: POLY-Ave,
		Ionosphere: POLY-Ave
		}
	\end{axis}
	\begin{semilogyaxis}[name = computationtime,
	    title = {Computation time},
	    at={($(rdaplot.south)-(0cm,2cm)$)},anchor=north,
	    width = \linewidth,
	    xlabel = {training set size},
	    ylabel = {median computation time (sec)},
	    legend style = {font=\tiny, draw=none},
	    legend columns = 3,
	    ymax = 4000,
	    ]
	    \addplot[red,mark=o] table[x=training_ratio,
	    y=rda.caret.cv5.time.median] {\sonarresults}; 
	    \addplot[red,mark=square] table[x=training_ratio,
	    y=rda.caret.cv5.time.median] {\vowelresults};
	    \addplot[red,mark=triangle*] table[x=training_ratio,
	    y=rda.caret.cv5.time.median] {\ionosphereresults};

	    \addplot[blue,mark=o] table[x=training_ratio,
	    y=rda.caret.cv10.time.median] {\sonarresults};
	    \addplot[blue,mark=square] table[x=training_ratio,
	    y=rda.caret.cv10.time.median] {\vowelresults};
	    \addplot[blue,mark=triangle*] table[x=training_ratio,
	    y=rda.caret.cv10.time.median] {\ionosphereresults};

	    \addplot[cyan,mark=o] table[x=training_ratio,
	    y=rda.klaR.cv5.time.median] {\sonarresults}; 
	    \addplot[cyan,mark=square] table[x=training_ratio,
	    y=rda.klaR.cv5.time.median] {\vowelresults}; 
	    \addplot[cyan,mark=triangle*] table[x=training_ratio,
	    y=rda.klaR.cv5.time.median] {\ionosphereresults}; 

	    \addplot[violet,mark=o] table[x=training_ratio,
	    y=rda.klaR.cv10.time.median] {\sonarresults};
	    \addplot[violet,mark=square] table[x=training_ratio,
	    y=rda.klaR.cv10.time.median] {\vowelresults};
	    \addplot[violet,mark=triangle*] table[x=training_ratio,
	    y=rda.klaR.cv10.time.median] {\ionosphereresults};

	    \addplot[black,mark=o] table[x=training_ratio,
	    y=rda.pool.time.median] {\sonarresults};
	    \addplot[black,mark=square] table[x=training_ratio,
	    y=rda.pool.time.median] {\vowelresults};
	    \addplot[black,mark=triangle*] table[x=training_ratio,
	    y=rda.pool.time.median] {\ionosphereresults};
	    \legend{
		Sonar: 5-CV,
		Vowel: 5-CV,
		Ionosphere: 5-CV,
		Sonar: 10-CV,
		Vowel: 10-CV,
		Ionosphere: 10-CV,
		Sonar: 5-CV-NM,
		Vowel: 5-CV-NM,
		Ionosphere: 5-CV-NM,
		Sonar: 10-CV-NM,
		Vowel: 10-CV-NM,
		Ionosphere: 10-CV-NM,
		Sonar: POLY-Ave,
		Vowel: POLY-Ave,
		Ionosphere: POLY-Ave}
	\end{semilogyaxis}
    \end{tikzpicture}
    \caption{\small Classification accuracy and computation time of the
    different methods for choosing the tuning
    parameters.}\label{fig:classification}
\end{figure}
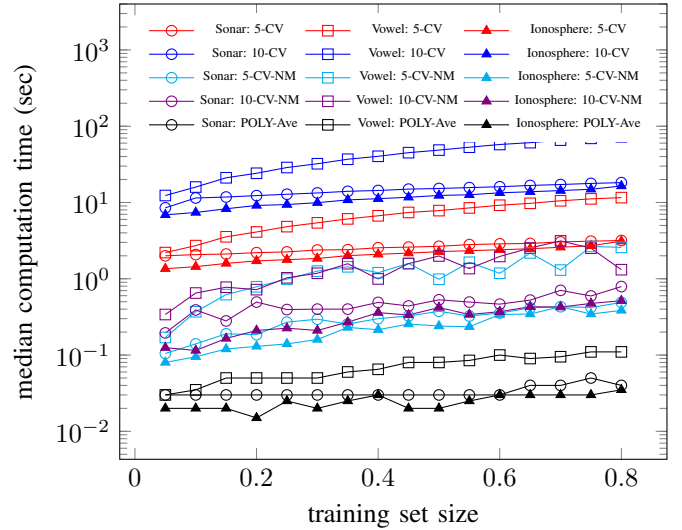

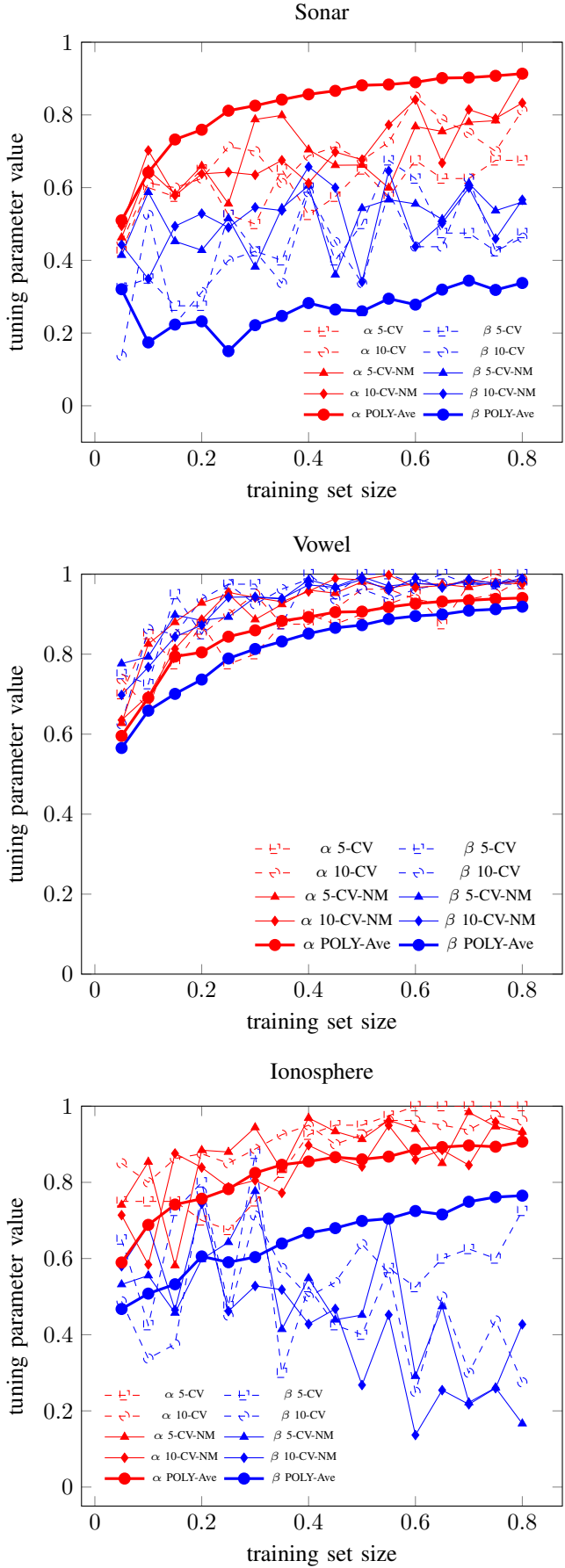
\begin{figure}
    \begin{tikzpicture}
	\begin{axis}[
		name = tuningsonar,
		width = \linewidth,
		xlabel = {training set size},
		ylabel = {tuning parameter value},
		legend pos = south east,
		legend style={font=\tiny, draw=none},
		title = Sonar,
		legend columns = 2,
		ymax = 1,
		ymin = -0.1,
		]
		\addplot[red,mark=square,dashed] table[x=training_ratio,
		y=alpha.caret.cv5.tl5] {\sonarresults};
		\addplot[blue,mark=square,dashed] table[x=training_ratio,
		y=beta.caret.cv5.tl5] {\sonarresults};

		\addplot[red,mark=o,dashed] table[x=training_ratio,
		y=alpha.caret.cv10.tl9] {\sonarresults};
		\addplot[blue,mark=o,dashed] table[x=training_ratio,
		y=beta.caret.cv10.tl9] {\sonarresults};

		\addplot[red,mark=triangle*] table[x=training_ratio,
		y=alpha.klaR.cv5.NelderMead] {\sonarresults};
		\addplot[blue,mark=triangle*] table[x=training_ratio,
		y=beta.klaR.cv5.NelderMead] {\sonarresults};

		\addplot[red,mark=diamond*] table[x=training_ratio,
		y=alpha.klaR.cv10.NelderMead] {\sonarresults};
		\addplot[blue,mark=diamond*] table[x=training_ratio,
		y=beta.klaR.cv10.NelderMead] {\sonarresults};

		\addplot[red,mark=*,very thick] table[x=training_ratio,
		y=alpha.poolave] {\sonarresults};
		\addplot[blue,mark=*,very thick] table[x=training_ratio,
		y=beta.poolave] {\sonarresults};
		\legend{
		    $\alpha$ 5-CV,
		    $\beta$ 5-CV,
		    $\alpha$ 10-CV,
		    $\beta$ 10-CV,
		    $\alpha$ 5-CV-NM,
		    $\beta$ 5-CV-NM,
		    $\alpha$ 10-CV-NM,
		    $\beta$ 10-CV-NM,
		    $\alpha$ POLY-Ave,
		    $\beta$ POLY-Ave}
	\end{axis}
	\begin{axis}[
		name = tuningvowel,
		at={($(tuningsonar.south)-(0cm,2cm)$)},
		anchor=north,
		width = \linewidth,
		xlabel = {training set size},
		ylabel = {tuning parameter value},
		legend pos = south east,
		legend style={font=\scriptsize, draw=none},
		title = Vowel,
		legend columns = 2,
		ymax = 1,
		ymin = 0,
		]

		\addplot[red,mark=square,dashed] table[x=training_ratio,
		y=alpha.caret.cv5.tl5] {\vowelresults};
		\addplot[blue,mark=square,dashed] table[x=training_ratio,
		y=beta.caret.cv5.tl5] {\vowelresults};

		\addplot[red,mark=o,dashed] table[x=training_ratio,
		y=alpha.caret.cv10.tl9] {\vowelresults};
		\addplot[blue,mark=o,dashed] table[x=training_ratio,
		y=beta.caret.cv10.tl9] {\vowelresults};

		\addplot[red,mark=triangle*] table[x=training_ratio,
		y=alpha.klaR.cv5.NelderMead] {\vowelresults};
		\addplot[blue,mark=triangle*] table[x=training_ratio,
		y=beta.klaR.cv5.NelderMead] {\vowelresults};

		\addplot[red,mark=diamond*] table[x=training_ratio,
		y=alpha.klaR.cv10.NelderMead] {\vowelresults};
		\addplot[blue,mark=diamond*] table[x=training_ratio,
		y=beta.klaR.cv10.NelderMead] {\vowelresults};

		\addplot[red,mark=*,very thick] table[x=training_ratio,
		y=alpha.poolave] {\vowelresults};
		\addplot[blue,mark=*,very thick] table[x=training_ratio,
		y=beta.poolave] {\vowelresults};

		\legend{
		    $\alpha$ 5-CV,
		    $\beta$ 5-CV,
		    $\alpha$ 10-CV,
		    $\beta$ 10-CV,
		    $\alpha$ 5-CV-NM,
		    $\beta$ 5-CV-NM,
		    $\alpha$ 10-CV-NM,
		    $\beta$ 10-CV-NM,
		    $\alpha$ POLY-Ave,
		    $\beta$ POLY-Ave}
	\end{axis}
	\begin{axis}[
		name = tuningionosphere,
		at={($(tuningvowel.south)-(0cm,2cm)$)},anchor=north,
		width = \linewidth,
		xlabel = {training set size},
		ylabel = {tuning parameter value},
		legend pos = south west,
		legend style={font=\tiny, draw=none},
		title = Ionosphere,
		legend columns = 2,
		ymax = 1,
		ymin = -0.05,
		]

		\addplot[red,mark=square,dashed] table[x=training_ratio,
		y=alpha.caret.cv5.tl5] {\ionosphereresults};
		\addplot[blue,mark=square,dashed] table[x=training_ratio,
		y=beta.caret.cv5.tl5] {\ionosphereresults};

		\addplot[red,mark=o,dashed] table[x=training_ratio,
		y=alpha.caret.cv10.tl9] {\ionosphereresults};
		\addplot[blue,mark=o,dashed] table[x=training_ratio,
		y=beta.caret.cv10.tl9] {\ionosphereresults};

		\addplot[red,mark=triangle*] table[x=training_ratio,
		y=alpha.klaR.cv5.NelderMead] {\ionosphereresults};
		\addplot[blue,mark=triangle*] table[x=training_ratio,
		y=beta.klaR.cv5.NelderMead] {\ionosphereresults};

		\addplot[red,mark=diamond*] table[x=training_ratio,
		y=alpha.klaR.cv10.NelderMead] {\ionosphereresults};
		\addplot[blue,mark=diamond*] table[x=training_ratio,
		y=beta.klaR.cv10.NelderMead] {\ionosphereresults};

		\addplot[red,mark=*,very thick] table[x=training_ratio,
		y=alpha.poolave] {\ionosphereresults};
		\addplot[blue,mark=*,very thick] table[x=training_ratio,
		y=beta.poolave] {\ionosphereresults};

		\legend{
		    $\alpha$ 5-CV,
		    $\beta$ 5-CV,
		    $\alpha$ 10-CV,
		    $\beta$ 10-CV,
		    $\alpha$ 5-CV-NM,
		    $\beta$ 5-CV-NM,
		    $\alpha$ 10-CV-NM,
		    $\beta$ 10-CV-NM,
		    $\alpha$ POLY-Ave,
		    $\beta$ POLY-Ave}
	\end{axis}
    \end{tikzpicture}
    \caption{\small The estimated tuning parameters for the data
    sets.}\label{fig:tuningparameters}
\end{figure}

\section{Conclusions}\label{sec:conclusion}
Two regularized sample covariance matrix estimators specified in
\eqref{eq:RSCMRDA} and in \eqref{eq:RSCMstreamlined} for multiclass problems
were considered in this work and their theoretically optimal (in terms of MSE)
class-specific tuning parameters were derived in Theorem~\ref{thm:poly} and
Theorem~\ref{thm:polysimple}, respectively. Since the optimal tuning
parameters depend on unknown scalar parameters in \eqref{eq:unknown_params}, a
method for their estimation was proposed. The usefulness of the method was
supported by the conducted numerical simulations, which demonstrated very good
MSE performance when compared to similar methods.  By averaging the
class-specific tuning parameters, the method was applied for choosing the tuning
parameters in an RDA classification framework. The proposed approach had
comparable classification performance to cross-validation on each of the three
tested real data sets, but had significantly faster computation time. The codes
for the proposed methods in Matlab, R, and Python programming languages are
available at \url{https://github.com/EliasRaninen}.

\appendices
\section{The MSE and optimal tuning parameters}\label{app:MSE}
\subsection{Notation and useful identities}
Let $\A$ and $\B$ be symmetric positive semidefinite matrices of same size.
Using the notation
\begin{align*}
    \I_\A = \frac{\tr(\A)}{p} \I
    \quad \text{and} \quad
    \A^\I = \A - \I_\A,
\end{align*}
we have the following identities:
\begin{itemize}
    \item $\tr(\I_\A) = \tr(\A)$
    \item $\tr(\A^\I) = 0$
    \item $\I_{\A+\B} = \I_\A + \I_\B$
    \item $(\A + \B)^\I = \A^\I + \B^\I$
    \item $\ip{\I_\A}{\I_\B} = \ip{\I_\A}{\B} = \ip{\A}{\I_\B} =
	p^{-1}\tr(\A)\tr(\B)$
    \item $\ip{\A^\I}{\I_\B} = 0$
    \item $\ip{\A^\I}{\B^\I} = \ip{\A}{\B} - \ip{\I_\A}{\I_\B}$
    \item $\ip{\A^\I}{\cdot} = \textstyle \ip{\A}{\cdot} - \ip{\I_\A}{\cdot}$
    \item $\fn{\A^\I}^2 = \fn{\A}^2 - \fn{\I_\A}^2$.
\end{itemize}
Because of linearity, we also have $\E[\I_{\A}] = \I_{\E[\A]}$ and $\E[\A^\I] =
\E[\A]^\I$. Regarding expressions involving the pooled SCM $\S$, we have
\begin{itemize}
    \item
	$ \fn{\S}^2 
	= \sum_i \pi_i^2 \|\S_i\|_\Fro^2 + \sum_{i\neq j} \pi_i \pi_j
	\ip{\S_i}{\S_j}$
    \item
	$\fn{\I_\S}^2 
	= \sum_i \pi_i^2 \|\I_{\S_i}\|_\Fro^2 + \sum_{i\neq j} \pi_i \pi_j
	\ip{\I_{\S_i}}{\I_{\S_j}}$
    \item
	$\ip{\cdot}{\S} = \sum_{j} \pi_j \ip{\cdot}{\S_j}$
    \item
	$\ip{\cdot}{\I_\S} = \sum_{j} \pi_j \ip{\cdot}{\I_{\S_j}}$.
\end{itemize}

\subsection{MSE of the estimator}\label{app:MSEoftheestimator}
The estimator $\hat \M_k (\alpha, \beta)$ in~\eqref{eq:RSCMRDA} can be
rewritten as
\begin{align*}
    \hat \M_k (\alpha, \beta) 
    = \alpha \beta {(\S_k - \S)}^\I + \alpha \S^\I + \beta \I_{(\S_k - \S)} +
    \I_\S.
\end{align*}
The squared error $\| \hat \M_k(\alpha,\beta) - \M_k \|^2_\Fro$ is
\begin{align*}
    &\fn{
	\alpha \beta {(\S_k - \S)}^\I
	+ \alpha \S^\I 
	+ \beta \I_{(\S_k - \S)} + \I_\S
	- \M_k
	}^2
	\\
	&=
	\alpha^2 \beta^2 \tilde C_{22}
	+ \alpha^2 \beta \tilde C_{21}
	+ \alpha \beta^2 \tilde C_{12}
	+ \alpha^2 \tilde C_{20}
	+ \beta^2 \tilde C_{02}
	\\&\qquad
	+ \alpha \beta \tilde C_{11}
	+ \alpha \tilde C_{10}
	+ \beta \tilde C_{01}
	+ \tilde C_{00},
\end{align*}
where
\begin{equation*}
    \begin{array}{ll}
	\tilde C_{22} = \fn{\S_k^\I - \S^\I}^2
	&
	\tilde C_{21} = 2\ip{\S_k^\I-\S^\I}{\S^\I}
	\\
	\tilde C_{12} = 0
	&
	\tilde C_{20} = \fn{\S^\I}^2
	\\
	\tilde C_{02} = \fn{\I_{\S_k} - \I_\S}^2
	&
	\tilde C_{11} = - 2\ip{\S_k^\I - \S^\I}{\M_k}
	\\
	\tilde C_{10} = -2\ip{\S^\I}{\M_k}
	&
	\tilde C_{01} = 2\ip{\I_{\S_k}-\I_\S}{\I_\S - \M_k}
	\\
	\tilde C_{00} = \fn{\I_\S - \M_k}^2.
    \end{array}
\end{equation*}
Taking the expectation of the squared error gives the MSE, with the
coefficients $C_{mn} = \E[\tilde C_{mn}]$, where $m,n \in \{0,1,2\}$.

For example, 
$C_{20} = \E[\tilde C_{20}] = \E[\|\S^\I\|_\Fro^2]
= \E[\|\S\|_\Fro^2] - \E[\|\I_\S\|_\Fro^2]$, where 
$\E[\|\S\|_\Fro^2] = \sum_i \pi_i^2 \E[\|\S_i\|_\Fro^2] + \sum_{i \neq
j} \pi_i \pi_j \ip{\M_i}{\M_j}$ and
$\E[\|\I_\S\|_\Fro^2] = \sum_i \pi_i^2 \E[\|\I_{\S_i}\|_\Fro^2] + \sum_{i \neq
j} \pi_i \pi_j \ip{\I_{\M_i}}{\I_{\M_j}}$. The estimation of these terms is
explained in Section~\ref{sec:estimationofparameters}.

The optimal tuning parameters can be solved exactly as follows. Let $L_k =
\MSE(\hat\M_k)$. The gradient equations of the MSE can be written as
\begin{align*}
    \left\{
	\begin{array}{ll}
	    \partial_\alpha L_k &=
	    2\alpha \beta^2 C_{22} 
	    + 2\alpha \beta C_{21}
	    + 2\alpha C_{20} 
	    + \beta C_{11}
	    + C_{10}
	    \\
	    \partial_\beta L_k &=
	    2\alpha^2 \beta C_{22} 
	    + \alpha^2 C_{21}
	    + 2\beta C_{02}
	    + \alpha C_{11}
	    + C_{01}.
	\end{array}
	\right.
\end{align*}
Set $\partial_\alpha L_k = 0$. Then, solving for $\alpha$ gives
\begin{align}
    \alpha &= -\frac{1}{2}
    \frac{\beta C_{11} + C_{10}}{\beta^2 C_{22} + \beta C_{21} + C_{20}}
    \label{eq:alpha}
\end{align}
when $\beta^2 C_{22} + \beta C_{21} + C_{20} \neq 0$, which holds a.s. for
continuous distributions (see~Appendix~\ref{app:biconvex}). By
substituting~\eqref{eq:alpha} into the equation $\partial_\beta L_k=0$, after
some algebra, we obtain a rational function. The zeros of the rational function
can be computed from the roots of its numerator, which is a quintic (fifth
order) polynomial in $\beta$.
To this end, general polynomial solvers can be used. The critical points of the
MSE are obtained by substituting the roots into~\eqref{eq:alpha}. The optimal
tuning parameters then correspond to the critical point, which yields the
minimum estimated MSE. If this critical point is not in the feasible set $[0,1]
\times [0,1]$, the other critical points and the boundaries need to be
considered, i.e., the special cases C1--C4 (see
Section~\ref{sec:derivationalphabeta}). For this purpose the
equations~\eqref{eq:alphagivenbeta} and~\eqref{eq:betagivenalpha} of
Theorem~\ref{thm:biconvex} can be used.

In practise, it may be simpler to find an approximately optimal tuning parameter
pair by forming a two-dimensional grid of $(\alpha,\beta)$ and choosing the
point with the smallest estimated MSE as explained in
Subsection~\ref{sec:computationoftuningparameters}. The approximate solution
can then further be finetuned by iterating~\eqref{eq:alphagivenbeta}
and~\eqref{eq:betagivenalpha} (see Appendix~\ref{app:biconvex}).

\subsection{Alternate convex minimization}\label{app:biconvex}
By considering one of the tuning parameters $\alpha$ or $\beta$ fixed, the
optimization of the remaining one is a convex problem. This is known as
biconvexity. To prove this, we show that the second derivatives of the MSE
$L_k$ are positive (implying strict convexity). Indeed, we have
\begin{align*}
    &\left\{
	\begin{array}{ll} 
	    \partial_\alpha^2 L_k = 2\beta^2 C_{22} + 2\beta C_{21} + 2
	    C_{20} > 0
	    \\
	    \partial_\beta^2 L_k = 2\alpha^2 C_{22} + 2 C_{02} > 0.
	\end{array}
	\right.
\end{align*}
The first equation is an upward opening quadratic function in $\beta$. Hence, it
is positive if its corresponding discriminant function is negative (there are no
real roots). Indeed, the discriminant, $C_{21}^2 - 4 C_{22} C_{20}$, is
negative, i.e., 
\begin{align*}
    (2\E[\ip{\S_k^\I - \S^\I}{\S^\I}])^2
    - 4 \E[\|\S_k^\I - \S^\I\|_\Fro^2]\E[\|\S^\I\|_\Fro^2] < 0 ~\text{a.s.},
\end{align*}
which follows from the Cauchy-Schwartz inequality. The second equation is also
positive since $C_{22}$ and $C_{02}$ are a.s. both positive.

Having concluded that the problem is biconvex, the (unconstrained) solution for
$\beta$ given a fixed value of $\alpha$ is
\begin{align*}
    \beta = -\frac{1}{2}\frac{\alpha^2 C_{21} + \alpha C_{11} 
    + C_{01}}{\alpha^2 C_{22} + C_{02}}.
\end{align*}
The corresponding (unconstrained) solution for $\alpha$ given a fixed value of
$\beta$ was already given in~\eqref{eq:alpha}. To guarantee that $\alpha,\beta
\in [0,1]$, a projection to the feasible set by the clip function $[\cdot]_0^1$
has to be applied. Hence, we obtain~\eqref{eq:alphagivenbeta}
and~\eqref{eq:betagivenalpha}. 

Lastly, we make some remarks about the convergence of sequentially
iterating~\eqref{eq:alphagivenbeta} and~\eqref{eq:betagivenalpha}. Let
$\alpha^{(i)}$ and $\beta^{(i)}$ denote the tuning parameter values at the $i$th
iteration. Denote the MSE function to be minimized as $L_k : X \times Y \to
\real$, where $X = [0,1]$ and $Y = [0,1]$. Since the MSE $L_k$ is bounded from
below and the sequence $\{L_k(\alpha^{(i)},\beta^{(i)})\}_{i \in \mathbb N}$ is
monotonically decreasing, it converges~\cite[Theorem
4.5]{gorskiBiconvexSetsOptimization2007}. Then, since $L_k$ is continuous, $X$
and $Y$ are closed sets, and both~\eqref{eq:alphagivenbeta}
and~\eqref{eq:betagivenalpha} have unique solutions (due to strict convexity) in
the compact set $[0,1]$, by~\cite[Theorem
4.9]{gorskiBiconvexSetsOptimization2007}, we have $\lim_{i \to \infty} \|
(\alpha^{(i+1)}, \beta^{(i+1)}) - (\alpha^{(i)}, \beta^{(i)}) \| = 0$.
Furthermore, since $X$ and $Y$ are subsets of $\mathbb R$, the algorithm is
coordinate-wise and the sequence $\{(\alpha^{(i)},\beta^{(i)})\}_{i\in \mathbb
N}$ converges~\cite[pp. 398]{gorskiBiconvexSetsOptimization2007}. If the
accumulation point of the sequence $\{(\alpha^{(i)}, \beta^{(i)})\}_{i\in\mathbb
N}$ is in the interior of $[0,1]\times [0,1]$, then it is a stationary
point~\cite[Corollary 4.10]{gorskiBiconvexSetsOptimization2007}. However, in
theory, the stationary point can be a global minimum, local minimum, or a saddle
point. Therefore, it is important to select the starting point for the
iterations carefully.

\section{Streamlined analytical estimator}\label{app:streamlined}
Let $\T \in \{\S, \S_k\}$, then the analytical estimator is
\begin{align*}
    \tilde \M_k(\alpha,\beta) 
    &= \alpha \beta (\S_k - \S) + \alpha (\S - \I_\T) + \I_\T.
\end{align*}
The squared error is
\begin{align*}
    &\fn{\alpha \beta (\S_k - \S) + \alpha (\S - \I_\T) + \I_\T - \M_k}^2
    \\&=
    \alpha^2 \beta^2 \tilde B_{22} + \alpha^2 \beta \tilde B_{21} 
    + \alpha^2 \tilde B_{20} + \alpha \beta \tilde B_{11} 
    + \alpha \tilde B_{10} + \tilde B_{00},
\end{align*}
where
\begin{equation*}
    \begin{array}{ll}
	\tilde B_{22} = \fn{\S_k - \S}^2
	&
	\tilde B_{21} = 2 \ip{\S_k - \S}{\S - \I_\T}
	\\
	\tilde B_{20} = \fn{\S - \I_\T}^2 
	&
	\tilde B_{11} = 2\ip{\S_k - \S}{\I_\T - \M_k}
	\\
	\tilde B_{10} = 2\ip{\S - \I_\T}{\I_\T - \M_k}
	&
	\tilde B_{00} = \fn{\I_\T - \M_k}^2.
    \end{array}
\end{equation*}
By taking the expectation, we get the MSE, i.e., $B_{mn} = \E[\tilde B_{mn}]$,
for $m,n \in \{0,1,2\}$. The partial derivatives of the MSE are
\begin{align*}
    &\left\{
	\begin{array}{ll}
	    \partial_\alpha L_k &= 
	    2 \alpha \beta^2 B_{22}
	    + 2 \alpha \beta B_{21}
	    + 2 \alpha B_{20}
	    + \beta B_{11}
	    + B_{10}
	    \\
	    \partial_\beta L_k &=
	    2 \alpha^2 \beta B_{22}
	    + \alpha^2 B_{21}
	    + \alpha B_{11}.
	\end{array}
	\right.
\end{align*}
By setting $\partial_\beta L_k = 0$, it is easy to see that $(\alpha = 0,
\beta = -B_{10}/B_{11})$ is a critical point. Note, however, that when $\alpha
= 0$, the estimator does not depend on $\beta$. Then, if we assume $\alpha
\neq 0$, solving for $\beta$ yields
\begin{align*}
    \beta = -\frac{1}{2} \frac{\alpha B_{21} + B_{11}}{\alpha B_{22}}.
\end{align*}
Substituting this expression into $\partial_\alpha L_k =0$ and solving for
$\alpha$ yields
\begin{equation}
    \alpha = \frac{2 B_{10} B_{22} - B_{11}B_{21}}{B_{21}^2 - 4 B_{20}B_{22}}.
    \label{eq:alphasimple}
\end{equation}
Substituting this back to the equation for $\beta$ gives
\begin{equation}
    \beta = \frac{2B_{11} B_{20} - B_{10}B_{21}}{2B_{10}B_{22} - B_{11}B_{21}}.
    \label{eq:betasimple}
\end{equation}
The Hessian matrix is
\begin{equation*}
    \mat H = 
    \begin{pmatrix}
	L_{\alpha\alpha} & L_{\alpha\beta} \\
	L_{\beta\alpha} & L_{\beta\beta}
    \end{pmatrix},
\end{equation*}
where the second partial derivatives are
\begin{align*}
    L_{\alpha\alpha} &= \partial_\alpha^2 L_k
    =
    2 \beta^2 B_{22} + 2 \beta B_{21} + 2 B_{20}
    \\
    L_{\beta\beta} &= \partial_\beta^2 L_k
    =
    2 \alpha^2 B_{22} 
    \\
    L_{\beta\alpha} &= L_{\alpha\beta} = \partial_\alpha \partial_\beta 
    L_k
    =
    4 \alpha \beta B_{22} + 2\alpha B_{21} + B_{11}.
\end{align*}
The determinant of the Hessian matrix is
\begin{align*}
    \det(\mat H) =
    L_{\alpha\alpha} L_{\beta\beta} - L_{\alpha\beta}^2.
\end{align*}
For the critical point $(\alpha=0,\beta=-B_{10}/B_{11})$, we have
$L_{\beta\beta}=0$ and $L_{\alpha\beta}=B_{11}$. The determinant of the
Hessian is then $-B_{11}^2 < 0$. This implies that the eigenvalues of the
Hessian matrix have different signs, and hence, the corresponding critical
point is a saddle point. Regarding the other critical point defined
by~\eqref{eq:alphasimple} and~\eqref{eq:betasimple}, after some manipulations,
one has that the determinant of the Hessian is
\begin{align*}
    \frac{(2B_{10}B_{22} - B_{11}B_{21})^2}{4B_{20}B_{22} - B_{21}^2},
\end{align*}
which is positive since by applying the Cauchy-Schwartz inequality
\begin{align*}
    0 < B_{21}^2 &= (2 \E[\ip{\S_k - \S}{\S - \I_\T}])^2
    \\&
    < 4 \E[\fn{\S_k - \S}^2] \E[\fn{\S - \I_\T}^2]
    = 4 B_{22} B_{20},
\end{align*}
which holds with strict inequalities a.s. for any continuous distribution.
Furthermore, as also $L_{\beta\beta} > 0$ a.s. (and hence $L_{\alpha\alpha} >
0$ a.s.), the second critical point is a local minimum.

\section*{Acknowledgment}
The authors would like to thank Dr. Tuomas Aittom\"aki for discussions that helped in
improving this paper.

\ifCLASSOPTIONcaptionsoff
\newpage
\fi

\renewcommand*{\bibfont}{\small}
\printbibliography

\end{document}